\documentclass[11pt,oneside,letterpaper]{article}
\usepackage{graphicx,amsmath,amsfonts}
\usepackage{setspace}
\usepackage{amssymb}
\usepackage{amsmath}
\usepackage{amsfonts}
\usepackage{mathtools}
\usepackage{graphicx}
\usepackage{enumerate}
\usepackage{color}
\usepackage{framed}
\usepackage[margin=1in]{geometry}
\usepackage{natbib}
\usepackage[pdfpagelabels]{hyperref}
\usepackage[capitalize]{cleveref}
\usepackage{amsthm}
\usepackage{caption}
\usepackage{subcaption}

\newtheorem{definition}{Definition}
\newtheorem{proposition}{Proposition}
\newtheorem{theorem}{Theorem}
\newtheorem{lemma}{Lemma}
\newtheorem{remark}{Remark}
\newtheorem{example}{Example}

\DeclareMathAlphabet{\mathbfit}{OML}{cmm}{b}{it}
\newcommand{\suchthat}{\;\ifnum\currentgrouptype=16 \middle\fi|\;}

\newcommand{\myref}[2]{\hyperref[#2]{$#1$\ref*{#2}}}

\newcommand{\todoIfTime}[1]{}
\newcommand{\todoLater}[1]{}
\newcommand{\todoBetterRevise}[1]{}
\newcommand{\todoRevise}[1]{}
\newcommand{\todoBetterAdd}[1]{}
\newcommand{\todoBetter}[1]{}

\let\originaleqref\eqref
\renewcommand{\eqref}{Eq.~\originaleqref}
\begin{document}

\title{Informational Braess' Paradox: The Effect of Information on Traffic
Congestion\thanks{%
We thank anonymous referees as well as participants at several seminars and conferences for useful suggestions and comments.}}
\author{Daron Acemoglu \thanks{%
Department of Economics, Massachusetts Institute of Technology, Cambridge,
MA, 02139 . \texttt{daron@mit.edu}} \and Ali Makhdoumi \thanks{%
Department of Electrical Engineering and Computer Science, Massachusetts
Institute of Technology, Cambridge, MA, 02139 \texttt{makhdoum@mit.edu}}
\and Azarakhsh Malekian \thanks{%
Rotman School of Management, University of Toronto, ON, M5S 3E6 \texttt{%
azarakhsh.malekian@rotman.utoronto.ca}} \and Asu Ozdaglar \thanks{%
Department of Electrical Engineering and Computer Science, Massachusetts
Institute of Technology, Cambridge, MA, 02139 \texttt{asuman@mit.edu}}}
\date{}
\maketitle

\begin{abstract}
To systematically study the implications of additional information about
routes provided to certain users (e.g., via GPS-based route guidance
systems), we introduce a new class of congestion games in which users have
differing information sets about the available edges and can only use routes
consisting of edges in their information set. After defining the notion of
Information Constrained Wardrop Equilibrium (ICWE) for this class of
congestion games and studying its basic properties, we turn to our main
focus: whether additional information can be harmful (in the sense of
generating greater equilibrium costs/delays). We formulate this question in
the form of Informational Braess' Paradox (IBP), which extends the classic
Braess' Paradox in traffic equilibria, and asks whether users receiving
additional information can become worse off. We provide a comprehensive
answer to this question showing that in any network in the series of
linearly independent (SLI) class, which is a strict subset of
series-parallel networks, IBP cannot occur, and in any network that is not in
the SLI class, there exists a configuration of edge-specific cost functions
for which IBP will occur. In the process, we establish several properties of
the SLI class of networks, which include the characterization of
the complement of the SLI class in terms of embedding a specific set of networks, and also an algorithm which determines whether a graph is SLI in
linear time. We further prove that the worst-case inefficiency performance
of ICWE is no worse than the standard Wardrop equilibrium.
\end{abstract}

\maketitle

\section{Introduction}

The advent of GPS-based route guidance systems, such as Waze or Google maps,
promises a better traffic experience to its users, as it can inform them
about routes that they were not aware of or help them choose dynamically
between routes depending on recent levels of congestion. Though other
drivers might plausibly suffer increased congestion as the routes they were
using become more congested due to this reallocation of traffic, or certain
residents may experience elevated noise levels in their side streets, it is
generally presumed that the users of these systems (and perhaps society as a
whole) will benefit. In this paper, we present a framework for
systematically analyzing how changes in the information sets of users in a
traffic network (e.g., due to route guidance systems) impact the traffic
equilibrium, and show the conditions under which even those with access to
additional information may suffer greater congestion.

Our formal model is a version of the well-known congestion games, augmented
with multiple types of users (drivers), each with a different information
set about the available edges in the network. These different information
sets represent the differing knowledge of drivers about the road network,
which may result from their past experiences, from inputs from their social
network, or from the different route guidance systems they might rely on. A
user can only utilize a route (path between origin and destination)
consisting of edges belonging to her information set. Each edge is endowed with a latency/cost function representing costs due to congestion. We generalize the
classic notion of \emph{Wardrop equilibrium} (\cite{wardrop1952some,  beckmann1956studies} and \cite{schmeidler1973equilibrium}), where each (non-atomic) user takes the level of
congestion on all edges as given and chooses a route with minimum cost (defined as the summation of costs of edges on the route). Our notion of \emph{Information Constrained Wardrop Equilibrium (ICWE)},
 also imposes the same equilibrium condition as Wardrop equilibrium,
but only for routes that are contained in the information set of each type
of user.

After establishing the existence and essential uniqueness of ICWE and
characterizing its main properties for networks with a single
origin-destination pair (an assumption we impose for simplicity and later relax), we turn to
our key question of whether expanding the information sets of some group of
users can make them worse off --- in the sense of increasing the level of
congestion they suffer in equilibrium. For this purpose, we define the
notion of \emph{Informational Braess' Paradox (IBP)}, designating the
possibility that users with expanded information sets experience greater
equilibrium cost. We then provide a tight characterization of when IBP
is and is not possible in a traffic network.

Our main result is that IBP does not occur if and only if the network
is \emph{series of linearly independent (SLI)}. More
specifically, this result means that in an SLI network, IBP can never occur,
ensuring that users with expanded information set always benefit from their
additional information. Conversely, if the network is not SLI, then
there exists a configuration of latency/cost functions for edges for which
IBP will occur. To understand this result, let us consider what the relevant
class of networks comprises. The set of SLI networks is a subset of \emph{series-parallel} 
networks, which are those for which two routes never pass through any edge in opposite directions. An SLI network is obtained by joining
together a collection of \emph{linearly independent (LI)} networks in series. LI
networks are those in which each route includes at least one edge that is not part of any other route. The intuition of our main result is as follows. To show IBP does not occur in an SLI network, we show the result on each of its LI parts. A key property of LI networks used in proving our main result is that for a traffic network with multiple information types if we reduce the total traffic demand, then there exists a route with strictly less flow.  When some users have more information, they change their routing, redirecting it to some subnetwork $A$ of the original network from some other subnetwork $B$ (and since the original network is LI, both $A$ and $B$ are also LI). All else equal, this will increase flows in $A$ and decrease flows in $B$. By the key property of LI networks, this will reduce flows in some route in $B$ and since users with more information have access to routes in $B$, this rerouting cannot increase their costs. Other users adjust their routing by allocating flows away from $A$ (since flow in $A$ has increased), which again by the LI property of the subnetwork implies that costs of some routes in $A$ decrease, establishing \textquotedblleft if\textquotedblright\ part of our main result. 
 The \textquotedblleft only if\textquotedblright\ part is
proved by showing that every non-SLI network embeds one of the collection of networks, and we demonstrate constructively that each one of these networks generates IBP (for some configuration of costs).

 We should also note that, since SLI is a restrictive class of
networks, and few real-world networks would fall into this class, we take
this characterization to imply that IBP is difficult to rule out in practice, and thus
the new, highly-anticipated route guidance technologies may make traffic
problems worse.

Since the class of SLI networks plays a central role in our analysis, a natural question is whether identifying
SLI networks is straightforward. We answer this question by showing that
whether a given network is SLI or not can be determined in linear time. This
result is based on the algorithms for identifying series-parallel networks proposed by
\cite{valdes1979recognition, schoenmakers1995new}, and \cite{ eppstein1992parallel}.

If, rather than considering a general change of information sets, we
specialize the problem so that only one user type does not have complete information about the available set of routes and the change in question is to bring all users complete information,
then we show that an IBP is possible if and only if the network is not 
\emph{series-parallel}. It is intuitive that this class of networks is less
restrictive than SLI, since we are now considering a specific change in
information sets (thus making IBP less likely to occur).

Our main focus is on traffic networks with a single origin-destination pair for which we provide a full characterization of network topologies for occurrence of IBP. In Section \ref{sec:IBPmultipleODpairs}, we consider multiple origin-destination pairs and use our characterization to provide a sufficient condition on the network topology under which IBP does not occur.

Our notion of IBP closely relates to the classic Braess' Paradox (BP), introduced in \cite{braess1968paradoxon} and further studied in  \cite{murchland1970braess} and
\cite{arnott1994economics}, which considers whether an additional edge in the network can increase equilibrium cost. When BP occurs in a network, IBP with a single information type also occurs (since IBP with a single information type can be shown to be identical to BP). Various aspects of BP and congestion games in general
is studied in \cite{murchland1970braess, steinberg1983prevalence,
dafermos1984some, patriksson1994traffic, bottom1999investigation, jahn2005system, ordonez2010wardrop, meir2014playing,
nikolova2014burden, chen2015excluding}, and \cite{feldman2015convergence}. Our characterization of ICWE and IBP clarifies that
our notion is different and, at least mathematically, more general. This can
be seen readily from a comparison of our results to the most closely related
papers to ours in the literature, \cite{milchtaich2005topological,
milchtaich2006network}. The characterizations in \cite{milchtaich2006network} imply that BP can be ruled out in series-parallel networks. Since IBP is a
generalization of BP, it should occur in a wider class of networks, and this
is indeed what our result shows- indeed SLI is a strict
subset of series-parallel networks. This result also indicates that IBP is a
considerably more pervasive phenomenon than BP. Notably, the mathematical argument
for our key theorem is different from \cite{milchtaich2006network} due
to the key difficulty relative to BP that not all users have access to the
same set of edges, and thus changes in traffic that benefit some groups of
users might naturally harm others by increasing the congestion on the routes
that they were previously utilizing.

Issues related to Braess' Paradox arise not only in the context of models of
traffic, but in various models of communication, pricing and choice over
congested goods, and electrical circuits. See e.g., \cite{orda1993competitive}, 
\cite{korilis1997achieving}, \cite{kelly1998rate}, and \cite%
{low1999optimization} for communication networks; the classic works by \cite%
{pigoueconomics} and \cite{samuelson1952spatial} as well as more recent works by 
\cite{johari2003network}, \cite{acemoglu2007competition, ashlagi2009two} and 
\cite{perakis2004price} for related economic problems; \cite{frank1981braess}%
, \cite{cohen1991paradoxical}, and \cite{cohen1997congestion} for mechanical
systems and electrical circuits; and \cite{rosenthal1973class} and \cite%
{vetta2002nash} for general game-theoretic approaches. This observation
also implies that the results we present here are relevant beyond traffic
networks, in fact to any resource allocation problem over a network subject
to congestion considerations. As pointed out in \cite{newell1980traffic} and 
\cite{sheffi1985urban} the Braess' paradox and related inefficiencies are a
clear and present challenge to traffic engineers, who often try to restrict
travel choices to improve congestion (e.g., via systems such as ramp
metering on freeway entrances). 

Other inefficiencies created by providing more information in the context of traffic networks have been studied in \cite{mahmassani1984dynamic, ben1991dynamic, arnott1991does}, and \cite{liu2016effects}. In particular, \cite{arnott1991does} consider a model with atomic drivers in which users decide on their departure time and route choice. They show that providing imperfect information regarding capacity/delay of roads might be worse than providing no information.  
More broadly, inefficiencies created by providing more information in other contexts are studied in \cite{maheswaran2003nash,
sanghavi2004optimal, yang2005revenue, harel2014more, dughmi2014hardness}, and \cite{rogers2015inducing}, among others.

Because our analysis also presents \textquotedblleft price of
anarchy\textquotedblright\ type results, i.e., bounds on the overall level of
inefficiency that can occur in an ICWE, our paper is related to previous
work on the price of anarchy in congestion and related games started by seminal
works of \cite{koutsoupias1999worst} and \cite{roughgarden2002bad} and
followed by \cite{correa2004selfish, correa2005inefficiency}, and \cite%
{friedman2004genericity}, as well as more generally to the analysis of
equilibrium and inefficiency in the variants of this class of games,
including \cite{milchtaich2004random, milchtaich2004social}, \cite%
{acemoglu2007partially, mavronicolas2007congestion, nisan2007algorithmic,
arnott1994economics, lin2004stronger, meir2015playing}, and \cite{anshelevich2008price}.
Here, our result is that the presence of users with different information
sets does not change the worst-case inefficiency traffic equilibrium as
characterized, for example, in \cite{roughgarden2002bad}.

The rest of the paper is organized as follows. In Section \ref{sec:model}, we
introduce our model of traffic equilibrium with users that are heterogeneous
in terms of the information about routes/edges they have access to, and then
define the notion of Information Constrained Wardrop Equilibrium for this
setting. In Section \ref{sec:existence} we prove the existence and essential
uniqueness of Information Constrained Wardrop Equilibrium. Before moving to
our main focus, in Section \ref{sec:graphtheory} we review some
graph-theoretic notions about series-parallel and linearly independent
networks, and then introduce the class of series of linearly independent
networks and prove some basic properties of this class of networks, which are
then used in the rest of our analysis. Section \ref{sec:IBP} defines our
notion of Informational Braess' Paradox. Section \ref{sec:characIBP}
contains our main result, showing that Informational Braess' Paradox occurs
\textquotedblleft if and only if\textquotedblright\ the network is not in
the class of series of linearly independent networks. Section \ref{sec:POA}
characterizes the worst-case inefficiency of Information Constrained Wardrop
Equilibrium, and finally, Section \ref{sec:conclusion} concludes.
 All the omitted proofs are included in the Appendix.  

\section{Model}
\label{sec:model}
We first describe the environment and then introduce our notion of
Information Constrained Wardrop Equilibrium.

\subsection{Environment}

We consider an undirected multigraph without self-loops denoted by $G=(V, \mathcal{E}, f)$ with vertex set $V$, edge set $\mathcal{E}$, and a function $f: \mathcal{E} \to \{\{u, v\}, u, v \in V, u \neq  v\}$ that maps each edge to its end vertices. For the ease of notation we will refer to $G$ as $(V, \mathcal{E})$ and denote an edge $e$ with $f(e)=\{u, v\}$ by $e=(u, v)$. We use the terms node and vertex interchangeably. Each edge $e\in \mathcal{E}$ joins two (distinct) vertices $u $ and $v$, referred to as the \emph{end vertices}
of $e$. An edge $e$
and a vertex $v$ are said to be \emph{incident} to each other if $v$ is an
end vertex of $e$. A \emph{path} $p \in G$ of length $n$ ($n\geq 0$) is a sequence of edges $e_{1}\ldots e_{n}$ in $\mathcal{E}$ where $e_i$ and $e_{i+1}$ share a vertex. If an edge $e$ appears on a path $p$, we write $e \in p$. The first and last vertices of a path $p$ are called the
initial and terminal vertices of $p$, respectively. If $q$ is a path of the
form $e_{n+1}\dots e_{m}$, with the initial vertex the same as
the terminal vertex of $p$ but all the other vertices and edges of $q$ do not belong to $p
$, then $e_{1}\dots e_n e_{n+1}\dots e_{m}$ is
also a path, denoted by $p+q$. For a path $p$ and two nodes $v$ and $u$ on it, we denote the \emph{section} of path between $u$ and $v$ by $p_{uv}$. 

Throughout the paper, we focus on an undirected multigraph $G=(V,\mathcal{E})$ together with an ordered pair of distinct vertices, called terminals, an origin $O$ and a
destination $D$, referred to as a \emph{network}. A subnetwork of $G$ is defined as $(V^{\prime }, \mathcal{E}%
^{\prime })$, where $V^{\prime }\subseteq V$ and $\mathcal{E}^{\prime } \subseteq \mathcal{E}$ and for any $e=(u, v) \in \mathcal{E}^{\prime}$, we have $u, v \in V^{\prime}$. We assume that each vertex and edge belong to at
least one path between the initial vertex $O$ and the terminal vertex $D$.
This assumption is without loss of generality because the vertices and edges
that do not belong to any path from $O$ to $D$ are irrelevant for the
purpose of sending traffic from $O$ to $D$. Any path $r$ with
$O$ as the initial vertex and $D$ as the terminal vertex will be called a \emph{route}. The set of all
routes in a network is denoted by $\mathcal{R}$. 
%
%

We suppose there are $K \ge 1$ types of users (we use the terms users and
players interchangeably) and use the shorthand notation $[K]=\{1,\dots ,K\}$ to denote the set of types. Each type $i\in \lbrack K \rbrack$ has total \emph{traffic demand} $s_{i} \in \mathbb{R}^{+}$, and we denote the vector of traffic demands by $s_{1:K}=(s_{1},\dots ,s_{K})$.
For each type $i$, we use $\mathcal{E}_{i}\subseteq \mathcal{E}$
to denote the set of edges that type $i$ knows and $\mathcal{R}_i$ to denote the routes formed by edges in $\mathcal{E}_i$ (assumed non-empty). We refer to $\mathcal{E}_i$ or $\mathcal{R}_i$ as type $i$'s information set. We use $\mathcal{E}_{1:K}= \left(\mathcal{E}_1, \dots, \mathcal{E}_K \right)$ to denote the information sets of all types.

We use $f^{(i)}=(f_{r}^{(i)}~:~r\in \mathcal{R}_{i})$ to denote the flow
vector of type $i$, where for all $r\in \mathcal{R}_{i}$, $f_{r}^{(i)}\geq 0$
represents the amount of traffic (flow) that type $i$ sends on route $r$. We use 
$f^{(1:K)}=(f^{(1)},\dots ,f^{(K)})$ to denote the flow vector of all types.
Each edge of the network has a cost (latency) function $c_{e}:\mathbb{R}%
^{+}\rightarrow \mathbb{R}^{+}$ which is continuous, nonnegative, and
nondecreasing. We denote the set of all cost functions by $\mathbf{c}=\{c_{e}~:~e\in \mathcal{E}\}$.
For instance, if all the cost functions are affine functions, then for any $%
e\in \mathcal{E}$, we would have $c_{e}(x)=a_{e}x+b_{e}$, for some $%
a_{e},b_{e}\in \mathbb{R}^{+}$. We refer to $(G, \mathcal{E}_{1:K},
s_{1:K}, \mathbf{c})$ as a \emph{traffic network with multiple information types}.
A feasible flow is a flow vector $f^{(1:K)}=(f^{(1)}, \dots,f^{(K)})$ such
that for all $i \in [K]$, $f^{(i)}$ is a flow vector of type $i$, i.e., $f^{(i)}: \mathcal{R}_i \to \mathbb{R
}^+$ and $\sum_{r \in \mathcal{R}_i} f^{(i)}_r = s_i$. We denote the total flow on each route $r$ by $f_r$%
, i.e., $f_r=\sum_{i=1}^K f_r^{(i)}$.


\subsection{Information Constrained Wardrop Equilibrium}

The cost of a route $r$ with respect to a flow $(f^{(1)}, \dots, f^{(K)})$ is
the sum of the cost of the edges that belong to this route, i.e., $%
c_r(f^{(1:K)})= \sum_{e \in r} c_e(f_e)$, where $f_e$ denotes the amount of
traffic that passes through edge $e$, i.e., $f_e=\sum_{r \in \mathcal{R}~:~
e \in r} f_r$. 

We assume flows get allocated at equilibrium according to a ``constrained"
version of Wardrop's principle: flows of each user type are routed along
routes in her information set with minimal (and hence equal) cost. We next
formalize this equilibrium notion.


\begin{definition}[\textbf{Information Constrained Wardrop Equilibrium (ICWE)%
}]\label{def:ICWE}
\textup{\ A feasible flow $f^{(1:K)}=(f^{(1)},\dots ,f^{(K)})$ is an Information Constrained Wardrop Equilibrium (ICWE) if for
every $i \in [K]$ and every pair $r,\tilde{r}\in \mathcal{R}_{i}$ with $f_{r}^{(i)}>0$, we have
\begin{equation}  \label{eq:defCWE}
c_{r}(f^{(1:K)})\leq c_{\tilde{r}}(f^{(1:K)}).
\end{equation}%
This implies that all the routes in $\mathcal{R}_i$ with positive flow from type $i$ have the same
cost, which is smaller or equal to the cost of any other route in $\mathcal{R%
}_{i}$. \emph{The equilibrium cost of type $i$}, denoted by $c^{(i)}$, is then
given by the cost of any route in $\mathcal{R}_{i}$ with positive flow from
type $i$. Note that the Wardrop Equilibrium (WE) is a special case of this
definition for a traffic network with a single information type, i.e., $K=1$%
. }
\end{definition}
We next provide an example that illustrates this definition and how it
differs from the classic Wardrop Equilibrium. 
\begin{figure}[tbp]
\centering
\includegraphics[width=0.4\textwidth]{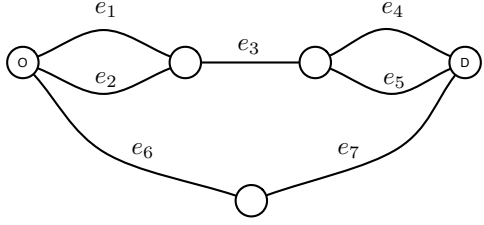}
\caption{Example of a network with edge cost functions given by $%
c_{e_{1}}(x)=c_{e_{4}}(x)=c_{e_{6}}(x)=x$ and  $%
c_{e_{2}}(x)=c_{e_{5}}(x)=c_{e_{7}}(x)=1+ax$ and $c_{e_{3}}=ax$ for some $a>0$.}
\label{fig:example11}
\end{figure}

\begin{example}
\label{ex:differentcosts} \textup{\ Consider the network $G=(V, \mathcal{E})$ given in Figure \ref
{fig:example11} with $s_{1}=s$, $s_{2}=1-s$, and the cost functions specified in Figure \ref
{fig:example11}. There are 5 different routes
from origin to destination, which we denote by $r_{1}=e_{1}e_{3}e_{4}$, $%
r_{2}=e_{1}e_{3}e_{5}$, $r_{3}=e_{2}e_{3}e_{4}$, $r_{4}=e_{2}e_{3}e_{5}$,
and $r_{5}=e_{6}e_{7}$. We let $\mathcal{E}_{1}=\mathcal{E}$ and $\mathcal{E}%
_{2}=\{e_{6},e_{7}\}$, which results
in $\mathcal{R}_{1}=\{r_{1},r_{2},r_{3},r_{4},r_{5}\}$ and $\mathcal{R}_{2}=\{r_{5}\}$, respectively.
\begin{itemize}
\item If $s \le \frac{2+a}{3+2 a}$, ICWE is $f^{(1)}_{r_1}= s$ and $f^{(2)}_{r_5}=1-s$. The equilibrium cost of type $1$ is $c^{(1)}=c_{r_1}(f^{(1:2)})=s+ a s +s = s(a+2)$. The equilibrium cost of type $2$ is $c^{(2)}=c_{r_5}(f^{(1:2)})=(1-s) + \left(1+ a (1- s) \right)=(1-s)(1+a)+1$. Hence, the equilibrium cost of type $1$ and type $2$ users need not be the same.
\item If $s > \frac{2+a}{3+2 a}$, ICWE is $f^{(1)}_{r_1}= \frac{2+a}{3+2 a} $, $f^{(1)}_{r_5}= s-\frac{2+a}{3+2 a} > 0$ and $f^{(2)}_{r_5}=1-s$, which give $c^{(1)}= c^{(2)}= \frac{(2+a)^2}{3+ 2a}$. This illustrates that when different types use a common route in an equilibrium, their equilibrium costs are the same.
\end{itemize}
}
\end{example}

\section{Existence of Information Constrained Wardrop Equilibrium}

\label{sec:existence}

In this section, we show that given a traffic network with multiple
information types $(G,\mathcal{E}_{1:K},s_{1:K},\mathbf{c})$, an ICWE always
exists and it is \textquotedblleft essentially\textquotedblright\ unique,
i.e., for each type, equilibrium cost is the same for all equilibria. Our
proof for existence and essential uniqueness of ICWE
relies on the following characterization, which is a straightforward extension of the
well-known optimization characterization of Wardrop Equilibrium (see \cite%
{beckmann1956studies} and \cite{smith1979existence}). 

\begin{proposition}
\label{pro:potential} A flow $f^{(1:K)}$ is an ICWE if and only if
it is a solution of the following optimization problem:
\begin{align}
& \min \sum_{e\in \mathcal{E}}\int_{0}^{f_{e}}c_{e}(z)dz  \notag
\label{eq:propotential} \\
& f_{e}=\sum_{i=1}^{K}\sum_{r\in \mathcal{R}_{i}~:~e\in r}f_{r}^{(i)},
\notag \\
& \sum_{r\in \mathcal{R}_{i}}f_{r}^{(i)}=s_{i},\text{ and }f_{r}^{(i)}\geq 0%
\text{ for all }r\in \mathcal{R}_{i}.
\end{align}%
We call $\sum_{e\in \mathcal{E}}\int_{0}^{f_{e}}c_{e}(z)dz$ the potential
function and denote it by $\Phi $.
\end{proposition}

Using the characterization of ICWE as the minimizer of a potential function,
we can now show the existence and essential uniqueness.

\begin{theorem}[\textbf{Existence and Uniqueness of ICWE}]
\label{thm:existsnceCWE} Let $(G, \mathcal{E}_{1:K}, s_{1:K}, \mathbf{c})$
be a traffic network with multiple information types.

\begin{itemize}
\item There exists an ICWE flow $f^{(1:K)}=(f^{(1)},\dots ,f^{(K)})$.

\item The ICWE is essentially unique in the sense that if $f^{(1:K)}$ and $%
\tilde{f}^{(1:K)}$ are both ICWE flows, then $%
c_{e}(f_{e})=c_{e}(\tilde{f}_{e})$ for every edge $e\in \mathcal{E}$.
\end{itemize}
\end{theorem}


\begin{remark}
\textup{
As shown in \cite{milchtaich2005topological, gairing2006routing}, and \cite{mavronicolas2007congestion} the essential uniqueness of equilibrium does not hold for multiple type  congestion games where different types of users have different cost functions for the same edge. This class of congestion games is also referred to as {\it player-specific congestion games}. Several conditions on the edge cost functions and network topology have been proposed to guarantee the existence of an essentially unique equilibrium (see \cite{konishi1997equilibria, voorneveld1999congestion, milchtaich2005topological, mavronicolas2007congestion, georgiou2009selfish}, and \cite{gairing2013congestion}). In particular, \cite{milchtaich2005topological} provides sufficient and necessary conditions on the network topology under which an essentially unique equilibrium exists. \cite{mavronicolas2007congestion} and \cite{georgiou2009selfish} show that when the edge costs are affine functions and differ by a player-specific additive constant, then an equilibrium exists. Our model is a special case of a player-specific congestion game in which the cost of an edge $e$ for a type $i$ user is $c_e(\cdot)$ if $e \in \mathcal{E}_i$ and $\infty$, otherwise. Therefore, the results of \cite{mavronicolas2007congestion} and \cite{georgiou2009selfish} can directly be used to establish the existence of an equilibrium in our model. For completeness, we provide an alternative proof of Theorem \ref{thm:existsnceCWE} in the Appendix \ref{app:proofs:sec:existence} based on the classical results of \cite{beckmann1956studies, schmeidler1973equilibrium, smith1979existence}, and \cite{milchtaich2000generic}. 
}
\end{remark}


Theorem \ref{thm:existsnceCWE} assumes that the cost functions are non-decreasing. If we strengthen this assumption to strictly increasing
cost functions,  then the results of \cite{roughgarden2002bad}, \cite{mavronicolas2007congestion}, and \cite{georgiou2009selfish}  show that the essential uniqueness result can be strengthened. In this case, the total flow on any edge at any equilibrium would be the same. 


\section{Some Graph-Theoretic Notions}

\label{sec:graphtheory}


In this section, we first present two classes of networks namely series-parallel and linearly independent networks
 which we use 
in our characterization of IBP. In preparation for our main graph-theoretic results, we also
present equivalent characterizations of these networks and delineate
the relations among them. Finally, we define a new class of networks termed series of linearly independent and present a characterization for it in terms of embedding of a few basic networks.

\begin{definition}[\textbf{Series-Parallel Network (SP)}]
\label{def:SP}
\textup{
A (two-terminal) network is called series-parallel if two
routes never pass through an edge in opposite directions. Equivalently, as was shown by \cite{RS42}, a network is
series-parallel if and only if
\begin{itemize}
\item[(i)] it comprises a single edge between $O$ and $D$, or 
\item [(ii)] it is constructed by connecting two series-parallel networks in series, i.e., by joining the destination of one series-parallel network with the origin of the other one, or 
\item [(iii)] it is constructed by connecting two series-parallel networks in parallel, i.e., by joining the origins and destinations of two series-parallel networks.
\end{itemize}
}
\end{definition}

As an example, the networks shown in Figure \ref{fig:fig11} and Figure \ref%
{fig:fig12} are series-parallel networks, while the network shown in Figure %
\ref{fig:fig13} is not series-parallel. The reason is that two routes $%
e_{1}e_{5}e_{4}$ and $e_{2}e_{5}e_{3}$ pass through the edge $e_{5}$ in
opposite directions.

An important subclass of series-parallel networks are \emph{linearly
independent} networks.

\begin{definition}[\textbf{Linearly Independent Network (LI)}]
\label{def:LI}
\textup{
A (two terminal) network is called linearly independent if each route has at least one edge that does not belong to any other route. Equivalently, as was shown by \cite{HL03}, a network is linearly independent if and only if
\begin{enumerate}
\item[(i)] it comprises a single edge between $O$ and $D$, or 
\item[(ii)] it is constructed by connecting a linearly independent network in series with a single edge network, or 
\item[(iii)] it is constructed by connecting two linearly independent networks in parallel. 
\end{enumerate}
}
\end{definition}

This class is termed linearly independent because of an algebraic characterization of the routes when viewed as vectors in the edge space. In particular, for any $r\in
\mathcal{R}$, let $\mathbf{v}_{r}\in \mathbb{F}_{2}^{|\mathcal{E}|}$ be $\mathbf{v}_{r}=(v_{r}^{1},\ldots ,v_{r}^{|\mathcal{E}|})$, where $v_{r}^{i}=1$ if $e_{i}\in r$ and $0$, otherwise. A network $G$ is LI if and only if the set of vectors $\{\mathbf{v}_{r}~:~r\in \mathcal{R}\}$ is linearly independent (see \cite{milchtaich2006network} and \cite{diestel2000graduate}).

As Definitions \ref{def:SP} and  \ref{def:LI} make it clear, the class of linearly
independent networks is a subset of the class of series-parallel networks.
An alternate characterization of linearly
independent and series-parallel networks  is based on the ``graph embedding'' notion, shown by \cite{duffin1965topology} and \cite{milchtaich2006network}, respectively.  We next define a graph embedding and then present these characterizations which will be used later in our analysis. 
\begin{figure}[t]
\centering
\begin{subfigure}[b]{0.25\textwidth}
 \includegraphics[width=\textwidth]{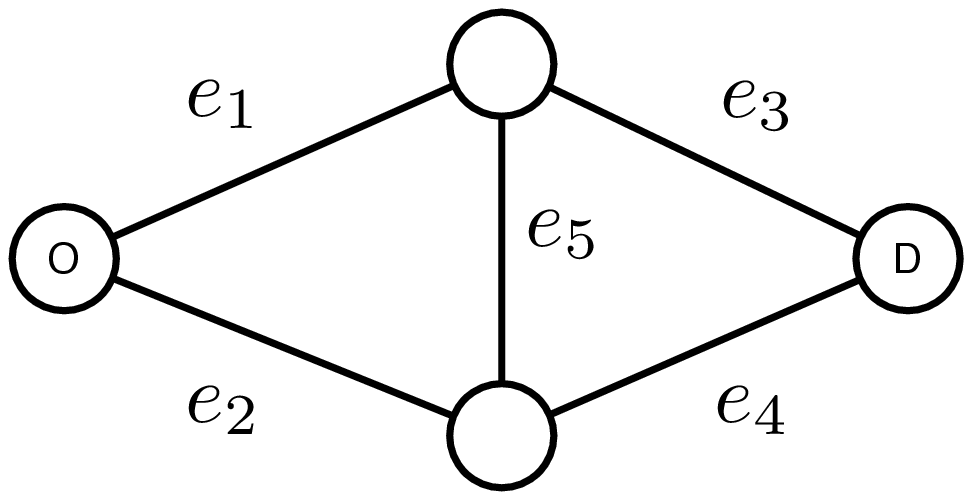}

                \caption{}
                                \label{fig:fig13}

        \end{subfigure}
~ 
\begin{subfigure}[b]{0.29\textwidth}
                \includegraphics[width=\textwidth]{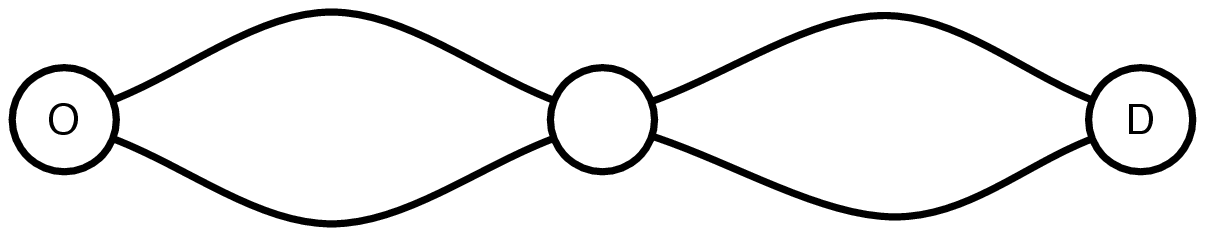}

                \caption{}
                                \label{fig:fig11}

        \end{subfigure}
~ 
\begin{subfigure}[b]{0.39\textwidth}
                \includegraphics[width=\textwidth]{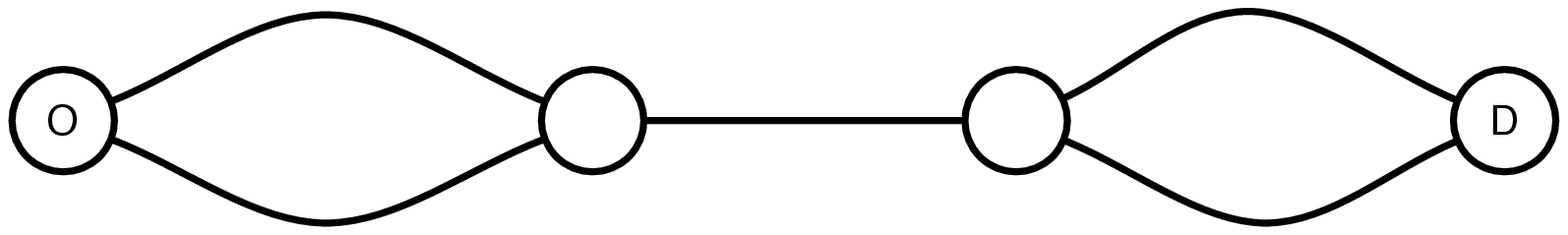}
                \caption{}
                \label{fig:fig12}
        \end{subfigure}
\caption{Networks that cannot be embedded in SP and LI networks: Network (a)
is not embedded in SP networks, Networks (a),(b), and (c) are not embedded in LI networks. }
\label{fig:threenetworks}
\end{figure}

\begin{definition}[\textbf{Embedding}]
\label{def:embedding} \textup{\ A network $H$ is embedded in the network $G$ if
we can start from $H$ and construct $G$ by applying the following steps in any order:
\begin{itemize}
\item [(i)] Divide an edge, i.e., replace an edge with two edges with a single common end
node.
\item [(ii)] Add an edge between two nodes.
\item [(iii)] Extend origin or destination by one edge.
\end{itemize}
}
\end{definition}


\begin{proposition}
\label{pro:alternatedef} 

\begin{itemize}
\item[(a)] [\cite{milchtaich2006network}] A network $G$ is LI if and only if none of the networks shown in Figure \ref
{fig:threenetworks} are embedded in it. Furthermore, a network $G$ is LI if and only if for every pair of
routes $r$ and $r^{\prime }$ and every vertex $v \neq O, D$ common to both routes,
either the section $r_{Ov}$ is equal to $r^{\prime }_{Ov}$, or $r_{vD}$ is
equal to $r^{\prime }_{vD}$.

\item[(b)] [\cite{duffin1965topology} and \cite{milchtaich2006network}] A network $G$ is SP if
and only if the network shown in Figure \ref{fig:fig13} is not embedded in it. Furthermore, a network $G$ is SP if and only if the vertices can be indexed in such a way that, along each route, they have increasing indices.
\end{itemize}
\end{proposition}

This proposition shows that series-parallel networks are those in which the network shown in Figure \ref%
{fig:fig13}, which is referred to as \emph{Wheatstone network} (see \cite{braess1968paradoxon}), is not embedded. LI networks, in addition, also exclude embeddings of series-parallel networks
that have routes that \textquotedblleft cross\textquotedblright  as
indicated in Figure \ref{fig:fig11} and Figure \ref{fig:fig12}.

We now introduce a new class of networks, which we refer to as \emph{series of linearly independent
networks} (SLI). 
\begin{definition}[\textbf{Series of Linearly Independent Network (SLI)}]\label{def:SLI}
\textup{\ A (two-terminal) network $G$ is called series of linearly independent if and only if
\begin{itemize}
\item [(i)] it comprises a single linearly independent network, or 
\item [(ii)] it is constructed by connecting two SLI networks in series.
\end{itemize}
A biconnected LI network is called an \emph{LI block}, where a graph is biconnected if it is connected and after removing any node and its incident edges the graph remains connected (see \citet[Chapter 3]{bondy1976graph}). Equivalently, a network $G$ is SLI if and only if it 
is constructed by connecting several LI blocks in series (see Appendix \ref{app:proofofDefSLI} for a formal proof). We refer to each of these blocks as an LI block of SLI network $G$. 
}
\end{definition}
\begin{figure}[t]
\centering
\begin{subfigure}[b]{0.19\textwidth}
               \includegraphics[width=\textwidth]{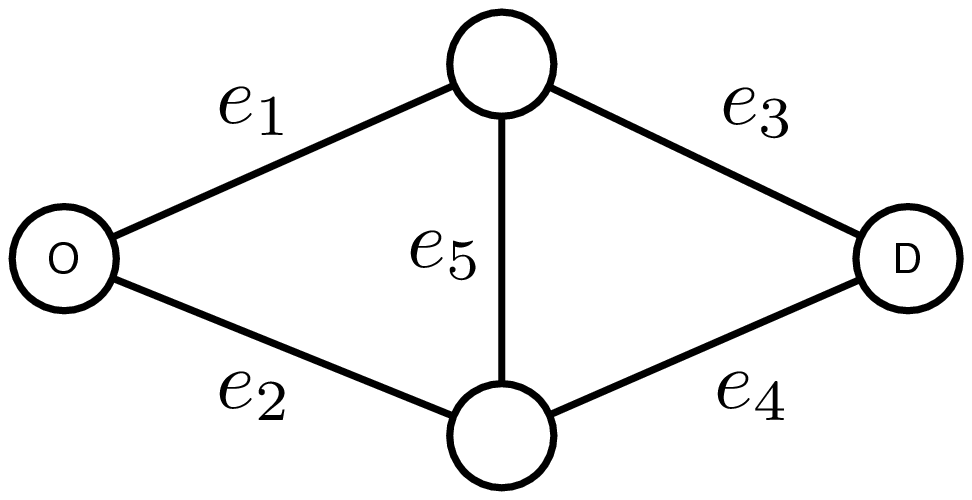}
                \caption{ }
                \label{fig:CE23}
        \end{subfigure}
~ 
\begin{subfigure}[b]{0.19\textwidth}
 \includegraphics[width=\textwidth]{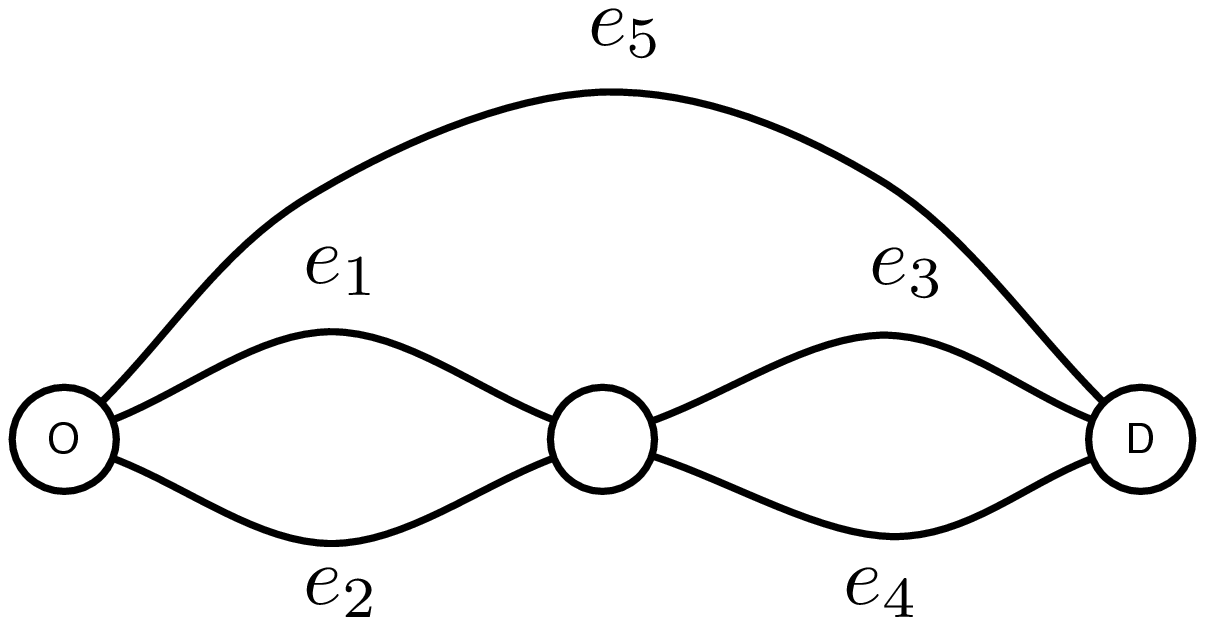}
                \caption{ }
                \label{fig:CE21}
                
        \end{subfigure}
~ 
\begin{subfigure}[b]{0.19\textwidth}
                \includegraphics[width=\textwidth]{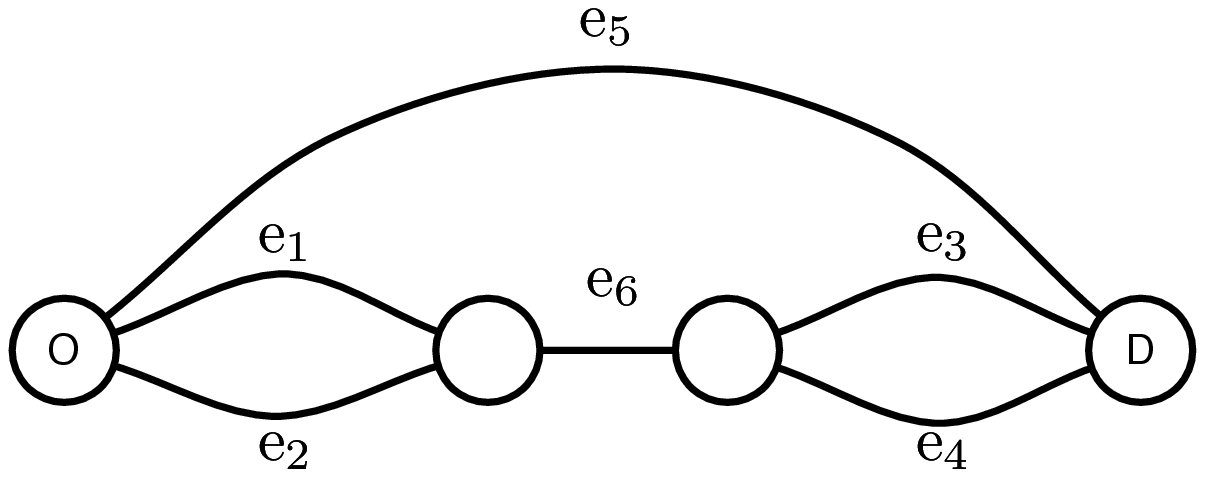}
                \caption{ }
                \label{fig:CE22}
        \end{subfigure}
~ 
\begin{subfigure}[b]{0.18\textwidth}
                \includegraphics[width=\textwidth]{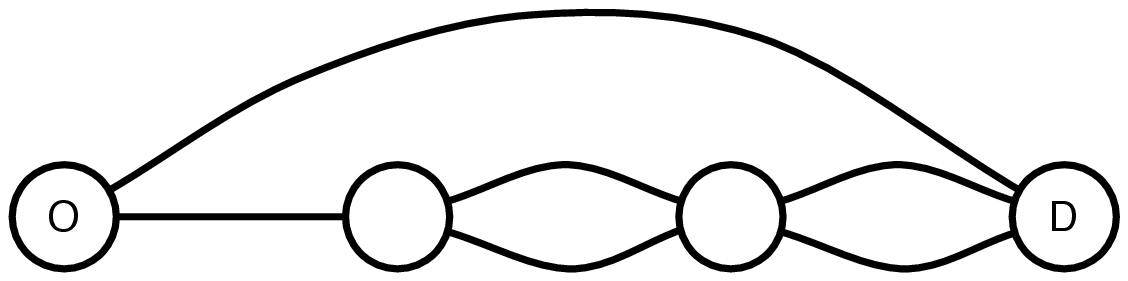}
                \caption{ }
                \label{fig:CE25}
        \end{subfigure}
~ 
\begin{subfigure}[b]{0.18\textwidth}
                \includegraphics[width=\textwidth]{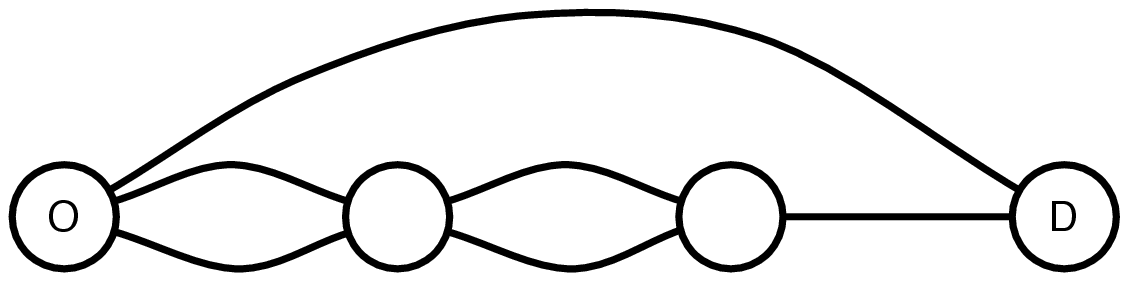}
                \caption{ }
                \label{fig:CE26}
        \end{subfigure}
~ 
\begin{subfigure}[b]{0.18\textwidth}
                \includegraphics[width=\textwidth]{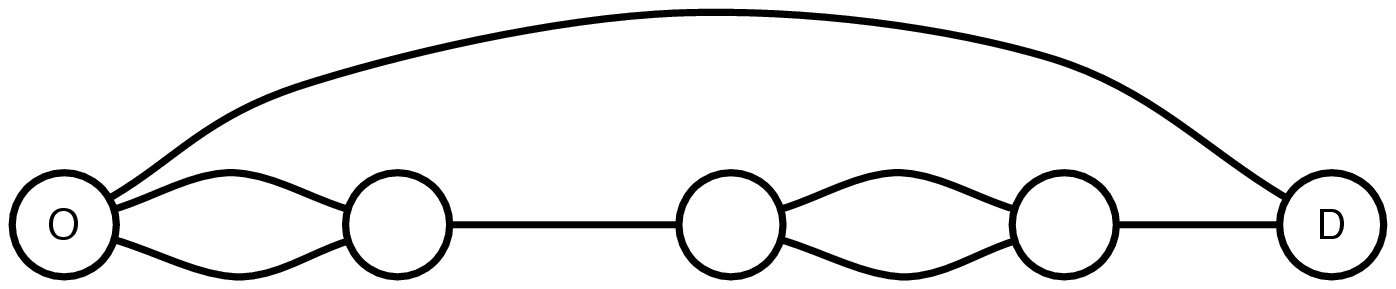}
                \caption{ }
                \label{fig:CE27}
        \end{subfigure}
~ 
\begin{subfigure}[b]{0.18\textwidth}
                \includegraphics[width=\textwidth]{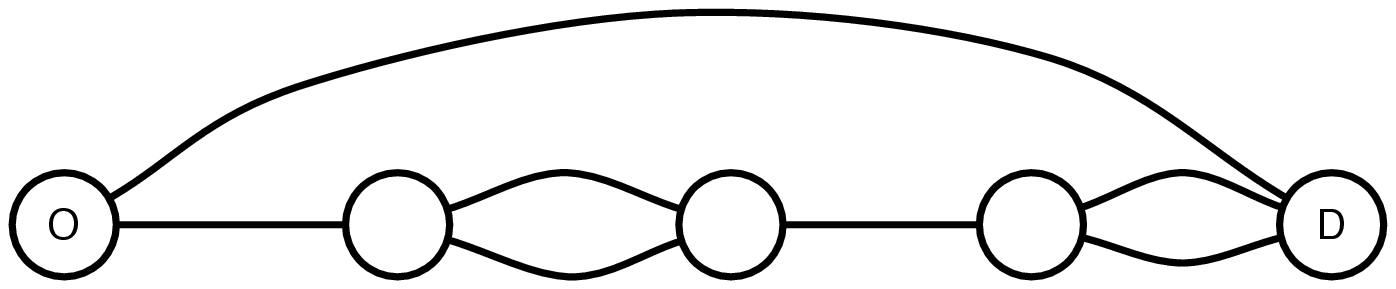}
                \caption{ }
                \label{fig:CE28}
        \end{subfigure}
~ 
\begin{subfigure}[b]{0.18\textwidth}
                \includegraphics[width=\textwidth]{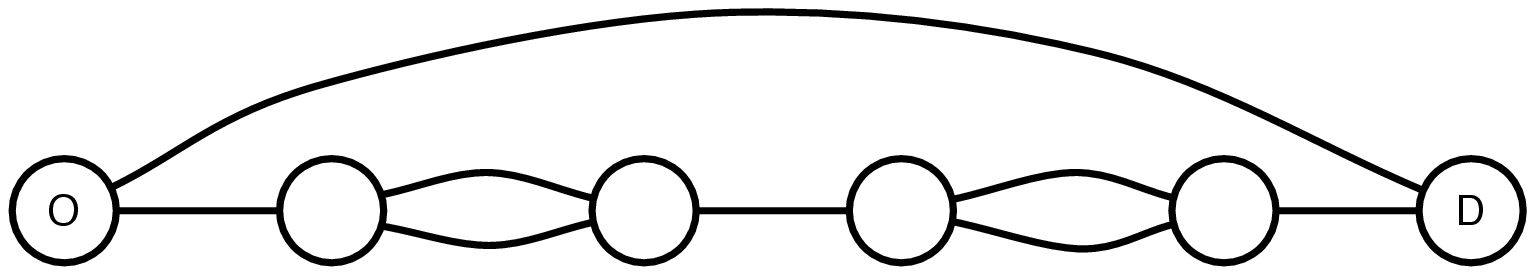}
                \caption{ }
                \label{fig:CE29}
        \end{subfigure}
~ 
\begin{subfigure}[b]{0.18\textwidth}
                \includegraphics[width=\textwidth]{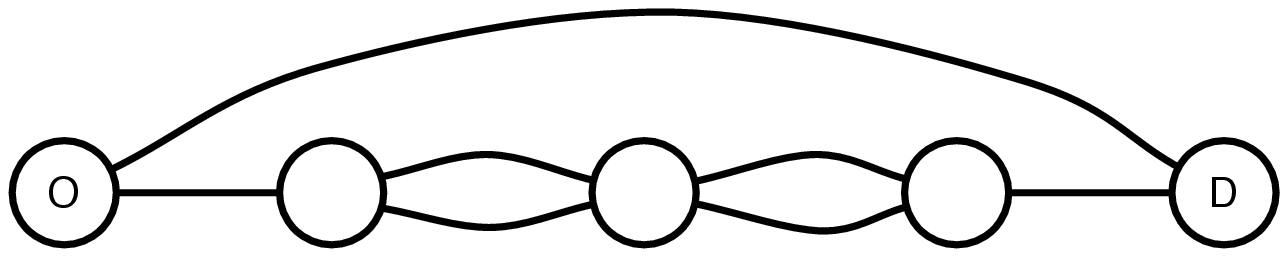}
                \caption{ }
                \label{fig:CE210}
        \end{subfigure}
\caption{Networks that cannot be embedded in SLI networks. }
\label{fig:CE2}
\end{figure}
We next provide a new characterization of SLI networks in terms of graph embedding using the
characterizations for SP and LI networks presented in Proposition \ref%
{pro:alternatedef}.

\begin{theorem}[\textbf{Characterization of SLI}]
\label{thm:SLIembedding} A network $G$ is SLI if and only
if none of the networks shown in Figure \ref{fig:CE2} are embedded in it.
\end{theorem}


%

The class of SLI networks is a subset of series-parallel networks and a superset
of linearly independent networks. 
This class plays an important role in
our characterization of networks that exhibit IBP. \cite{valdes1979recognition} provided an algorithm to determine whether a given network is SP in $O\left(|\mathcal{E}|+|V|\right)$ steps based on a tree decomposition of SP networks. This leads to the question whether one can find a linear time algorithm (i.e., linear in the number of vertices and edges) to recognize an SLI network. We next use the results of \cite{valdes1979recognition} to show that we can recognize whether a given network is SLI in linear time.  

\begin{proposition}
\label{pro:recognitionSLI} There exists an algorithm that can determine
whether a given network $G$ is SLI in $O\left( |\mathcal{E}|+|V| \right)$.
\end{proposition}



\section{Informational Braess' Paradox}\label{sec:IBP} 
We first present the classical Braess' Paradox (BP) which is defined for a
traffic network with single type of users with $\mathcal{E}_{1}=\mathcal{E}$%
, denoted by $(G,\mathcal{E}_{1},s_{1},\mathbf{c})$.

\begin{definition}[\textbf{Braess' Paradox (BP)}]
\textup{\ Consider a traffic network with single information type $(G,%
\mathcal{E}_{1},s_{1},\mathbf{c})$. BP occurs if there exists another set of
cost functions $\hat{\mathbf{c}}$ with $\hat{c}_{e}(x)\leq {c}_{e}(x)$ for
all $e\in \mathcal{E}$ and $x \in \mathbb{R}^+$, such that the
equilibrium cost of $(G,\mathcal{E}_{1},s_{1},\hat{\mathbf{c}})$
is strictly larger than the equilibrium cost of $(G,\mathcal{E}_{1},s_{1},{%
\mathbf{c}})$. }
\end{definition}

BP refers to an unexpected increase in equilibrium cost in response to a decrease in edge costs. We next
discuss the Informational Braess' Paradox (IBP), which arises when providing more
information to a subset of users in a traffic network increases those users' costs. 

\begin{definition}[\textbf{Informational Braess' Paradox (IBP)}]\label{def:IBP}
\textup{\ Consider a traffic network with multiple information types $(G,%
\mathcal{E}_{1:K},s_{1:K},\mathbf{c})$. IBP occurs if there exist expanded information sets $\tilde{\mathcal{E}}_{1:K}$ with $\mathcal{E}%
_{1}\subset \tilde{\mathcal{E}}_{1} \subseteq \mathcal{E}$ and $\tilde{\mathcal{E}}_{i}=\mathcal{E}%
_{i}$, for $i=2,\dots ,K$, such that the equilibrium cost of type $1$ in $(G,\tilde{\mathcal{E}}_{1:K},s_{1:K},\mathbf{c})$ is strictly larger
than the equilibrium cost of type $1$ in $(G,\mathcal{E}_{1:K},s_{1:K},
\mathbf{c})$. We denote the equilibrium cost of type $i\in [K]$ before and after the expansion of information sets by $c^{(i)}$ and $\tilde{c}^{(i)}$, respectively. }
\end{definition}

The choice of type $1$ users in this definition is without loss of
generality, i.e., we assume that the information set of only one type
expands and the information sets of the rest of the types remain the same. In comparing IBP to BP, first note that BP occurs in a network if and only if a special case of BP occurs in which we decrease the cost of one of the edges from infinity to its actual cost, i.e., equilibrium cost increases by adding a new edge to the network.  The \textquotedblleft if \textquotedblright part holds by definition and the  \textquotedblleft only if \textquotedblright part holds because the special case of BP occurs in Wheatstone network (as presented in Example \ref{example:notSLI}(a)) and Wheatstone network is embedded in any network that features BP as shown by \cite{milchtaich2005topological}. In light of this, it follows that the occurrence of IBP is a generalization of that of BP since addition of a new edge to the network can be viewed as expansion of the information set of a type to include that edge in a traffic network with single information type.

The next example shows that IBP occurs in all networks shown in Figure \ref{fig:CE2}, i.e.,  all the basic networks that are embedded in non-SLI networks. 

\begin{example}
\label{example:notSLI} \textup{ In this example we will show that for all
networks shown in Figure \ref{fig:CE2}, there exists an assignment of cost
functions along with information sets for which IBP occurs.
\begin{itemize}
\item [(a)] IBP occurs for Wheatstone network shown in Figure \ref{fig:CE23}.  This follows from the occurrence of BP on Wheatstone network as shown in \cite{braess1968paradoxon}. We will provide the example for the sake of completeness in Appendix \ref{app:exampleinfinitelymany}.
\item [(b)] Consider the network shown in Figure \ref{fig:CE21} with cost functions given by $c_{e_1} (x)= \frac{1}{2} x$, $c_{e_2} (x)= x+ \frac{3}{4}$, $c_{e_3} (x)= \frac{4}{3}x$, $c_{e_4}(x)=2$ and $c_{e_5}(x)=x$. The information sets are $\mathcal{E}_1=\{e_2, e_3, e_5\}$, $\mathcal{E}_2=\{e_1, e_4, e_5\}$, and $\tilde{\mathcal{E}}_1=\{e_1, e_2, e_3, e_5\}$. For $s_1=\frac{13}{4}$ and $s_2=1$, the equilibrium flows are
\begin{align*}
&f^{(2)}_{e_1 e_4}= 1, f^{(2)}_{e_5}= 0, f^{(1)}_{e_2 e_3}= \frac{3}{4}, f^{(1)}_{e_5}= \frac{10}{4},\\
&\tilde{f}^{(2)}_{e_1 e_4}= 0, \tilde{f}^{(2)}_{e_5}= 1, \tilde{f}^{(1)}_{e_2 e_3}=0, \tilde{f}^{(1)}_{e_1 e_3}=\frac{6}{4}, \tilde{f}^{(1)}_{e_5}= \frac{7}{4}.
\end{align*}
The resulting equilibrium costs are $c^{(1)}=c^{(2)} = \frac{10}{4}$ and $\tilde{c}^{(1)} = \tilde{c}^{(2)}= \frac{11}{4}$. Since $\tilde{c}^{(1)} > c^{(1)}$, IBP occurs in this network. The main intuition for this example is as follows. After adding $e_1$ to $\mathcal{E}_1$, type $1$ users will no longer use $e_2e_3$ and instead redirect part of their flow over $e_1e_3$. This in turn will increase the cost of $e_1 e_4$ for type $2$ users, and induce them to redirect all their flow from $e_1 e_4$ to $e_5$. In balancing the costs of $e_1e_3$ and $e_5$ for type $1$ users, their equilibrium cost goes up. 
\item [(c)] Finally, for the networks shown in Figures \ref{fig:CE22}, \ref{fig:CE25}, \ref{fig:CE26}, \ref{fig:CE27}, \ref{fig:CE28}, \ref{fig:CE29}, and \ref{fig:CE210}, IBP occurs if we use the same setting as part (b) and include extra edges in all information sets with zero cost. 
\end{itemize}
}
\end{example}
\begin{remark}\label{remark:expansionofexample} \textup{ In Appendix \ref{app:expansionofexample}, we show that Example \ref{example:notSLI}(b) is not degenerate and provide an infinite set of (affine) cost functions for which IBP occurs in this network. Similar to Example \ref{example:notSLI}(c), this argument extends to show that there are infinitely many cost functions for which IBP occurs in networks shown in Figures \ref{fig:CE22},$\dots$, \ref{fig:CE210}.  
Finally, for the network shown in Figure \ref{fig:CE23}, there are infinitely many cost functions for which BP occurs when edge $e_5$ is added, hence IBP occurs as well (see e.g. \cite{steinberg1983prevalence}). 
}
\end{remark}

In a seminal paper, \cite{milchtaich2006network} provided necessary and sufficient conditions on the network topology under which BP occurs. In particular, \cite{milchtaich2006network} showed that for a given traffic network with single information type
$(G,\mathcal{E}_{1},s_{1},\mathbf{c})$, BP does not occur if and only if $G$
is SP. That is, if $G$ is SP, then for any assignment of cost functions $
\mathbf{c}$ and traffic demand, BP does not occur, and  if $G$ is not SP, then there exists an assignment of cost
functions $\mathbf{c}$ for which BP occurs. 

We next investigate conditions on the topology of the network under which
IBP occurs. Similar to the characterization provided by \cite{milchtaich2006network}, we will identify classes of networks for which IBP does not
occur regardless of the cost functions of the edges. Since, as already
noted, IBP is a strict generalization of BP, we will see that IBP can occur
in a broader class of networks, underscoring the problem mentioned in the
introduction that IBP is likely to be a more pervasive problem.\textbf{\ } 

\section{Characterization of Informational Braess' Paradox}

\label{sec:characIBP}

In this section, after establishing the key lemmas which underpin the rest
of our analysis, we provide our main characterization of IBP. In Subsection \ref{sec:IBPres} we provide a characterization for IBP for a more restricted type of change in information sets. We conclude this section with a discussion of extensions of our results to multiple origin-destination pairs. 

\subsection{Three Key Lemmas}

The following lemmas identify properties of the traffic network
consisting of heterogeneous users over an LI network. 

\begin{lemma}\label{lemma:LI12}
\begin{itemize}
\item [(a)] Given an LI network $G$ , let $f^{(1:K)}$ and $\tilde{f}^{(1:K)}$ be two arbitrary
non-identical feasible flows for two traffic networks $(G,\mathcal{E}_{1:K},s_{1:K},\mathbf{c})$ and $(G,\mathcal{E}_{1:K},\tilde{s}_{1:K},\mathbf{c})$, respectively. If $\sum_{i=1}^{K}s_{i}\geq
\sum_{i=1}^{K}\tilde{s}_{i}$, then there exists a route $r$ such that $%
\sum_{i=1}^{K}f_{r}^{(i)}>\sum_{i=1}^{K}\tilde{f}_{r}^{(i)}$ and $f_{e}\geq
\tilde{f}_{e}$, for all $e\in r$.
\item [(b)] Given an LI network $G$,  let $c^{(i)}$ and $\tilde{c}^{(i)}$ denote
the equilibrium cost of type $i\in [K]$ users in traffic networks $(G,\mathcal{E}_{1:K},s_{1:K},\mathbf{c})$ and $(G,\tilde{\mathcal{E}}_{1:K},s_{1:K},\mathbf{c})$, respectively.
If  
$\mathcal{E}_{1}\subseteq \tilde{\mathcal{E}}_{1}$ and 
$\tilde{\mathcal{E}}_{i}=\mathcal{E}_{i}$, for  $i=2,\dots ,K$, then there exists
some $i\in \lbrack K \rbrack$ such that $\tilde{c}^{(i)}\leq c^{(i)}$.
\end{itemize}
\end{lemma}
This lemma directly follows from \citet[Lemma 5 and Theorem 3]{milchtaich2006network}. The first part of the this lemma shows that in an LI network, if the total traffic increases, then there exists at least one route whose flow strictly increases, and the flow on each of its edges weakly increases. The second part shows that in an LI network, if we expand the
information set of type $1$ users, then the equilibrium cost of at
least one of the types does not increase. In fact, a similar argument shows that even if we expand the information set of multiple types, then the equilibrium cost of at
least one of the types does not increase (see \citet[Theorem 3]{milchtaich2006network}). Note that this result is not sufficient for
establishing that IBP does not occur over LI networks because what we need to
establish is that it is the equilibrium cost of type $1$ users that does not
increase. 
For completeness, in Appendix \ref{app:Proof:lemma:LI12}, we show how this lemma follows from the results of \cite{milchtaich2006network}. 

The next lemma shows a property of equilibrium flows and equilibrium costs in a network which is the result of attaching two networks in series. We use the following definition to state the lemma. Suppose $f^{(1:K)}$ is a feasible flow for $(G, \mathcal{E}_{1:K}, \mathbf{s}_{1:K}, \mathbf{c})$ where $G$ is the result of attaching $G_1$ and $G_2$ in series. We denote the attaching point of $G_1$ and $G_2$ by $D_1$. The \emph{restriction} of $f^{(1:K)}$ to $G_1$ (similarly to $G_2$) is defined as $\bar{f}^{(1:K)}= (\bar{f}^{(1)}, \dots, \bar{f}^{(K)})$ where the flow of type $i$ on any route $\bar{r}$ in $G_1$ is the summation of the flows of type $i$ on all routes of $G$ which contain $\bar{r}$. Formally, for any $i \in [K]$ we have $\bar{f}^{(i)} (\bar{r})= \sum_{r \in \bar{\mathcal{R}}_i (\bar{r})} f^{(i)}(r)$, where $\bar{\mathcal{R}}_i (\bar{r})= \{ r \in \mathcal{R}_i~:~ r_{O D_1}= \bar{r} \}$. 
Note that $\bar{f}^{(1:K)}$ is a feasible flow on $G_1$.
\begin{lemma}\label{lem:sumflowseries}
\begin{itemize}
\item [(a)] If $G$ is the result of attaching two networks $G_1$ and $G_2$ in series, then the restriction of an equilibrium flow for $G$ to each of $G_1$ and $G_2$ is an equilibrium flow. 
\item [(b)] If $G$ is the result of attaching two networks $G_1$ and $G_2$ in series, then the equilibrium cost of any type on $G$ is the summation of the equilibrium costs of that type on $G_1$ and $G_2$.
\end{itemize}
\end{lemma}
The third lemma shows our key lemma that we will use in the proof of
Theorem \ref{thm:SLI}. Intuitively, this lemma states that in an LI network, if we decrease the traffic 
on one subset of routes $\mathcal{R}_A$ of the network and reroute it through the rest of the routes in the network, denoted by $\mathcal{R}_B=\mathcal{R}\setminus \mathcal{R}_A$, then the maximum cost improvement over all the routes in $\mathcal{R}_A$ cannot be smaller than the minimum cost improvement over all the routes in $\mathcal{R}_B$. This
result will enable us to establish that in an LI or SLI network, the
reallocation of traffic due to one type of users obtaining more information
cannot harm that type.

\begin{lemma}
\label{lem:maxminLI} Given an LI network $G$,  we let  
$\mathcal{R}_{A}, \mathcal{R}_{B} \neq \emptyset$ denote a partition of routes 
$\mathcal{R}$, i.e., $\mathcal{R}_{B} = \mathcal{R} \setminus \mathcal{R}_A$. We let $f^{(1:K)}$ and 
$\tilde{f}^{(1:K)}$ be two feasible flows for traffic networks 
$(G,\mathcal{E}_{1:K},s_{1:K},\mathbf{c})$ and $(G,\mathcal{E}_{1:K},\tilde{s}_{1:K},\mathbf{c})$, respectively. For these two flows, we let the traffic over $\mathcal{R}_A$ and $\mathcal{R}_B$ be 
 $s_{A}=\sum_{r\in \mathcal{R}_{A}}\sum_{i=1}^{K}f_{r}^{(i)}$, $\tilde{s}_{A}=\sum_{r\in \mathcal{R}_{A}}\sum_{i=1}^{K}\tilde{f}_{r}^{(i)}$, $s_{B}=\sum_{r\in \mathcal{R}_{B}}\sum_{i=1}^{K}f_{r}^{(i)}$, and $\tilde{s}_{B}=\sum_{r\in \mathcal{R}_{B}}\sum_{i=1}^{K}\tilde{f}_{r}^{(i)}$.
If $\tilde{s}_{A}\leq s_{A}$ and $\tilde{s}_{B}\geq s_{B}$, then we have
\begin{equation*}
\max_{r\in \mathcal{R}_{A}}\{ c_{r}-\tilde{c}_{r} \} \geq \min_{r\in \mathcal{R}_{B}} \{ c_{r}-\tilde{c}_{r} \},
\end{equation*}
where for any route $r$, $c_r$ and $\tilde{c}_r$ denote the cost of this route with flows $f^{(1:K)}$ and $\tilde{f}^{(1:K)}$, respectively. 
\end{lemma}

Before proving this lemma for a general LI network, let us show it for the
special case where $G$ consists of parallel edges from $O$ to $D$. In this
case $\mathcal{R}_{A}$ and $\mathcal{R}_{B}$ are two disjoint sets of edges
from $O$ to $D$. 
Since $\tilde{s}_{A}\leq s_{A}$, there exists an edge $e_{A}$ in $\mathcal{R}_{A}$ such that $\tilde{f}_{e_{A}}\leq f_{e_{A}}$. Similarly, since $\tilde{s}_{B}\geq s_{B}$, there exists $e_{B}\in \mathcal{R}_{B}$ such that $\tilde{f}_{e_{B}}\geq f_{e_{B}}$. Since the cost functions are nondecreasing, we have  
\begin{equation*}
\max_{r\in \mathcal{R}_{A}}\{ c_{r}-\tilde{c}_{r} \} \geq
c_{e_{A}}(f_{e_{A}})-c_{e_{A}}(\tilde{f}_{e_{A}})\geq 0\geq
c_{e_{B}}(f_{e_{B}})-c_{e_{B}}(\tilde{f}_{e_{B}})\geq \min_{r\in \mathcal{R}_{B}} \{ c_{r}-\tilde{c}_{r} \},
\end{equation*}%
which is the desired result. The proof for the general case is by induction
on the number of edges and is included next.
\begin{proof}
We first note a
consequence of Proposition \ref{pro:alternatedef}: 
\\\textbf{Claim 1:} If a network $G$ is LI then for any vertex $v$, either the sections from $O$ to $v$ of all routes that pass through $v$ (which consists of $v$ and all the vertices and edges preceding it on the route) are identical or the sections from $v$ to $D$ of all routes that pass through $v$ (which consists of $v$ and all the vertices and edges succeeding it on the route) are identical. Consider a route $r$ that passes through $v$. First, note that since $G$ is SP, part (b) of Proposition \ref{pro:alternatedef} implies that the only common node of $r_{Ov}$ and $r_{vD}$ is $v$. Also, the only common node of $r'_{Ov}$ and $r'_{vD}$ is $v$.  Claim 1 follows since if the contrary holds, then there exist two routes $r=r_{Ov}+ r_{vD}$ and $r'=r'_{Ov}+r'_{vD}$ with a common vertex $v$ such that $r_{Ov} \neq r'_{Ov}$ and $r_{vD} \neq r'_{vD}$. This contradicts the statement of part (a) of Proposition \ref{pro:alternatedef}.
\newline
We now prove Lemma \ref{lem:maxminLI} using induction on the number of edges. For a single edge it evidently holds. For a general 
LI network, we have the
following cases:
\begin{itemize}
\item There exists $r \in \mathcal{R}_A$ such that $c_r \ge \tilde{c}_r$ and $%
r^{\prime }\in \mathcal{R}_B$ such that $c_{r^{\prime }} \le \tilde{c}%
_{r^{\prime }}$. This leads to
\begin{align*}
\max_{r\in \mathcal{R}_A} \{ c_r-\tilde{c}_r \} \ge c_r-\tilde{c}_r \ge 0 \ge
c_{r^{\prime }}-\tilde{c}_{r^{\prime }}\ge \min_{r\in \mathcal{R}_B} \{ c_r-\tilde{c%
}_r \},
\end{align*}
which concludes the proof in this case.
\item For any $r\in \mathcal{R}_{A}$, we have $c_{r}<\tilde{c}_{r}$. We break the proof into three steps. 

\noindent \textbf{Step 1:} There exists a route $r \in \mathcal{R}_A$ and an edge $e \in r$ with the following properties: (i) The flow on $r$ from $\tilde{s}_A$ is less than or equal to the flow on $r$ from $s_{A}$. (ii) The flow on $e$ from $\tilde{s}_B$ is larger than the flow on $e$ from $s_{B}$ and the flow on $e$ from $\tilde{s}_A$ is less than or equal to the flow on $e$ from $s_{A}$. 

\noindent This step follows from invoking part (a) of Lemma \ref{lemma:LI12}. Since $
\tilde{s}_{A}\leq s_{A}$, using part (a) of Lemma \ref{lemma:LI12} there
exists a route $r\in \mathcal{R}_{A}$ such that the flow on each edge of $r$
from $\tilde{s}_{A}$ is less than or equal to the flow from $s_{A}$.
However, we know that the overall cost of any $r\in \mathcal{R}_{A}$ has gone
up, i.e., $\tilde{c}_{r}>c_{r}$. This implies that there exists an edge $e\in r$ such that  the flow from $\tilde{s}%
_{B}$ on $e$ is more than the flow from $s_{B}$ on $e$.

\noindent \textbf{Step 2:} Let $\mathcal{R}_{e}$  denote the set of routes using edge $e=(u_e,v_e)$ as defined in Step 1. $\mathcal{R}_e$ has the following properties: (i) Either there exists a vertex $D'\in V$ such that all routes $r\in \mathcal{R}_e$ have a common path from $O$ to $v_e$ and a common path from $D'$ to $D$, or there exists a vertex $O'\in V$ such that all routes $r\in \mathcal{R}_e$ have a common path from $u_e$ to $D$ and a common path from $O$ to $O'$. Without loss of generality, we assume it is the former case. (ii) There exists a subnetwork $G'$ with origin $O'=v_e$ and destination $D'$ such that for the restricted parts of $\mathcal{R}_A$ and $\mathcal{R}_B$ over $G'$, denoted by $\mathcal{R}'_A$ and $\mathcal{R}'_B$, if  we let $s'_A$, $\tilde{s}'_A$, $s'_B$, and $\tilde{s}'_B$ to denote the corresponding traffic demands on $\mathcal{R}'_A$ and $\mathcal{R}'_B$, then we have $\tilde{s}'_A  \le s'_A$ and $\tilde{s}'_B \ge s'_B$.

\noindent  Using Claim
1, for an edge $e=(u_e,v_e)$ either there is a unique path from $O$ to $v_e$, or there is a unique path from $v_e$ to $D$; we assume without loss of generality it is the former case. We let $D'$ to be the first node on route $r$ such that all routes in $\mathcal{R}_e$ coincide from $D'$ to $D$. Therefore, all routes $r\in \mathcal{R}_e$ have a common path from $O$ to $v_e$ and a common path from $D'$ to $D$, showing the first property.\\
We next  show that the subnetwork consisting of all routes from $v_e$ to $D'$, denoted by $G'$, satisfies the second property. 
To see this, note that the flows on $G'$ are only the ones that are passing through edge $e$. From Step 1, we know that the flow on $e$ from $\tilde{s}_B$ is larger than the flow on $e$ from $s_{B}$ and the flow on $e$ from $\tilde{s}_A$ is less than or equal to the flow on $e$ from $s_{A}$, showing the second property.  

\noindent \textbf{Step 3:}  Using steps 1 and 2, and induction hypothesis for $G'$, we will show that 
\begin{align*}
\max_{r \in \mathcal{R}_A} \{ c_r -\tilde{c}_r \} \ge \min_{r \in
\mathcal{R}_B} \{ c_r - \tilde{c}_r \}.
\end{align*}
First note that $\mathcal{R}_e \cap \mathcal{R}_A \neq \emptyset$ since $r\in \mathcal{R}_A$ and $e\in r$.  Furthermore, $\mathcal{R}_e \cap \mathcal{R}_B \neq \emptyset$, since as explained in Step 1, the flow on $e$ from $\tilde{s}_B$ is strictly positive. This in turn shows that $\mathcal{R}'_A$ and $\mathcal{R}'_{B}$ are nonempty. 
Using Step 2, all the conditions of Lemma \ref{lem:maxminLI} hold for subnetwork $G'$. Therefore, we can use the induction hypothesis for LI network $G'$ to obtain 
\begin{align*}
\max_{r \in \mathcal{R}^{\prime }_A} \{ c_r -\tilde{c}_r \} \ge \min_{r \in
\mathcal{R}^{\prime }_B} \{ c_r - \tilde{c}_r \}.
\end{align*}
Using Step 2, for all the routes in $\mathcal{R}_e$ the costs of going
from $O$ to $O'$, denoted by $c_{O' \to O}$ are the same. Similarly, the costs of all routes in $\mathcal{R}_e$ going from $D'$ to $D$, denoted by $c_{D' \to D}$ are the same. 
 Therefore, we have
\begin{align*}
& \max_{r \in \mathcal{R}_A} \{ c_r- \tilde{c}_r \} \ge \max_{r \in \mathcal{R}_A \cap \mathcal{R}_e} \{ c_r- \tilde{c}_r \} =  (c_{O \to O'}- \tilde{c}_{O
\to O'})+\max_{r \in \mathcal{R}^{\prime }_A} \{ c_r -\tilde{c}_r \} +(c_{D'
\to D}- \tilde{c}_{D' \to D}) \\
& \ge (c_{O \to O'}- \tilde{c}_{O \to O'})+ \min_{r \in \mathcal{R}^{\prime
}_B} \{ c_r - \tilde{c}_r \} + (c_{D' \to D}- \tilde{c}_{D' \to D})= \min_{r \in \mathcal{R}_B \cap \mathcal{R}_e} \{ c_r - \tilde{c}_r \}
\ge \min_{r \in \mathcal{R}_B} \{ c_r - \tilde{c}_r \} ,
\end{align*}
which concludes the proof in this case.
\item For any $r \in \mathcal{R}_B$, we have $c_r > \tilde{c}_r$. The proof of this case is similar to the previous case. We state the three steps without repeating the reasoning for each of them. 

\noindent \textbf{Step 1:} There exists a route $r \in \mathcal{R}_B$ and an edge $e \in r$ with the following properties: (i) The flow on $r$ from $\tilde{s}_B$ is larger than or equal to the flow on $r$ from $s_{B}$. (ii) The flow on $e$ from $\tilde{s}_A$ is smaller than the flow on $e$ from $s_{A}$ and the flow on $e$ from $\tilde{s}_B$ is larger than or equal to the flow on $e$ from $s_{B}$. 

\noindent \textbf{Step 2:} Let $\mathcal{R}_{e}$  denote the set of routes using edge $e=(u_e,v_e)$ as defined in Step 1. We have the following properties: (i) Either there exists a vertex $D'\in V$ such that all routes $r\in \mathcal{R}_e$ have a common path from $O$ to $v_e$ and a common path from $D'$ to $D$, or there exists a vertex $O'\in V$ such that all routes $r\in \mathcal{R}_e$ have a common path from $u_e$ to $D$ and a common path from $O$ to $O'$. Without loss of generality, we assume it is the former case. (ii) There exists a subnetwork $G'$ with origin $O'=v_e$ and destination $D'$ such that for the restricted parts of $\mathcal{R}_A$ and $\mathcal{R}_B$ over $G'$, denoted by $\mathcal{R}'_A$ and $\mathcal{R}'_B$, if  we let $s'_A$, $\tilde{s}'_A$, $s'_B$, and $\tilde{s}'_B$ to denote the corresponding traffic on $\mathcal{R}'_A$ and $\mathcal{R}'_B$, then we have $\tilde{s}'_A  \le s'_A$ and $\tilde{s}'_B \ge s'_B$.

\noindent \textbf{Step 3:} Again, using steps 1 and 2, and induction hypothesis for $G'$, we have
\begin{align*}
& \max_{r \in \mathcal{R}_A} \{ c_r- \tilde{c}_r \} \ge  \min_{r \in \mathcal{R}_B} \{ c_r - \tilde{c}_r \} ,
\end{align*}
which completes the proof.

\end{itemize}
\end{proof}

\subsection{Characterization of Informational Braess' Paradox}

We next present our main result, which states that IBP does not occur if
and only if the network is SLI. The idea of this
result, as already discussed in the Introduction, is the following. To show the \textquotedblleft if\textquotedblright\ part, we note that using Lemma \ref{lem:sumflowseries} it suffices to show IBP does not occur in LI networks. Consider an expansion of the information set of type $1$ and the new equilibrium flows. If the equilibrium cost of type $1$ increases, then ICWE definition implies the following:
\begin{itemize}
\item Consider all types with increased equilibrium costs (including type $1$). All routes used by these types (in the equilibrium before information expansion) have higher costs in the new equilibrium. 
\item Consider all types with decreased equilibrium costs. All routes used by these types (in the equilibrium after information expansion) have lower costs in the new equilibrium. 
\end{itemize}
Using these two claims, it follows that the total flow sent over the routes with higher costs is lower, and the total flow sent over the routes with lower costs is higher (see Figure \ref{fig:proofMultiple}). Since the network is LI, Lemma \ref{lem:maxminLI} leads to a contradiction. 
 The \textquotedblleft only
if\textquotedblright\ part holds because any non-SLI network embeds
one of the networks shown in Figure \ref{fig:CE2}, and an IBP can
be constructed for each of them (Example \ref{example:notSLI}) which then extends to an IBP for the non-SLI network. 
\begin{theorem}[\textbf{Characterization of IBP}]
\label{thm:SLI} IBP does not occur if and only if
$G$ is SLI. More specifically, we have the following.
\begin{itemize}
\item[(a)] If $G$ is SLI, for any traffic network $(G,\mathcal{E}_{1:K},s_{1:K},\mathbf{c})$ with arbitrary assignment of cost functions $\mathbf{c}$, $K$, traffic demands $s_{1:K}$, and information sets $\mathcal{E}_{1:K}$,  IBP does not occur.

\item[(b)] If $G$ is not SLI, there exists an assignment of cost functions $\mathbf{c}$, $K$, traffic demands $s_{1:K}$, and information sets $\mathcal{E}_{1:K}$ in which IBP
occurs.
\end{itemize}
\end{theorem}
\begin{proof}
\textbf{Part (a):} 
To reach a contradiction, suppose that $\tilde{c}^{(1)}>c^{(1)}$. By Definition \ref{def:SLI}, $G$ is obtained from attaching several LI blocks in series, denoted by $G_1, \dots, G_{N}$ for some $N \ge 1$. Using part (b) of Lemma \ref{lem:sumflowseries}, we have $\tilde{c}^{(1)}= \sum_{t=1}^N \tilde{c}^{(1)}_t >   \sum_{t=1}^N c^{(1)}_t  = c^{(1)}$, where $c^{(1)}_t$ denotes the equilibrium cost of type $1$ users in $G_t$. Therefore, there exists one LI block such as $j$ for which $\tilde{c}^{(1)}_j > c^{(1)}_j$. Also, using part(a) of Lemma \ref{lem:sumflowseries}, the restriction of equilibrium flows $f^{(1:K)}$ and $\tilde{f}^{(1:K)}$ to $G_j$ creates an equilibrium flow for this LI block. Therefore, IBP occurs in LI block $G_j$. In the rest of the proof of part (a),
we will assume IBP occurs in an LI block (and hence LI network) and reach a contradiction. We let $f^{(1:K)}$ and $\tilde{f}^{(1:K)}$ be the equilibrium flows before and after the information set expansion. Also, for any route $r \in \mathcal{R}$, we let $c_{r}$ and $\tilde{c}_r$ to denote the cost of route $r$ with flows $f^{(1:K)}$ and $\tilde{f}^{(1:K)}$, respectively. 
\newline
We partition the set $[K]$ into groups $A$ and $B$ as follows 
\begin{equation*}
A=\{i\in \lbrack k \rbrack~:~\tilde{c}^{(i)}>c^{(i)}\},
\end{equation*}%
and 
\begin{equation*}
B=\{i\in \lbrack k \rbrack~:~\tilde{c}^{(i)}\leq c^{(i)}\},
\end{equation*}%
i.e., set $A$ denotes all types with higher equilibrium cost in the game
with higher information, and set $B$ denotes the rest of the types.\newline
We also partition the routes of the network into two subsets $\mathcal{R}%
_{A} $ and $\mathcal{R}_{B}$, where 
\begin{equation*}
\mathcal{R}_{A}=\{r\in \mathcal{R}~:~\tilde{c}_{r}>c_{r}\},
\end{equation*}%
and 
\begin{equation*}
\mathcal{R}_{B}=\{r\in \mathcal{R}~:~\tilde{c}_{r}\leq c_{r}\},
\end{equation*}%
i.e., $\mathcal{R}_{A}$ denotes all routes that have higher costs in the
game with higher information, and $\mathcal{R}_{B}$ denotes the rest of the
routes. We show the following claims: \newline
\textbf{Claim 1:} For any type $i\in A$ and any route $r\in \mathcal{R}_{B}$, we have $%
f_{r}^{(i)}=0$, i.e., for a given type $i$, if the equilibrium cost increases in the game
with higher information, then the cost of all routes that type $i$ was
using (with strictly positive flow) also increases. This follows since if $r\not\in \mathcal{R}_{i}$, then $f_{r}^{(i)}=0$.
Otherwise, $r\in \mathcal{R}_{i}$ which implies $r\in \tilde{\mathcal{R}}%
_{i} $ as well, where $\tilde{\mathcal{R}}_i$ denotes the set of available routes to type $i$ in the expanded information set. Assuming $i\in A$ and  $r\in \mathcal{R}_B$, we have
\begin{equation*}
c_{r} \ {\geq } \ \tilde{c}_{r} \ {\geq } \ \tilde{c}^{(i)} \ {>} \ c^{(i)}, 
\end{equation*}%
where the first inequality follows from the definition of the set $\mathcal{R%
}_{B}$. The second inequality follows from the
definition of ICWE. The third inequality follows from the definition of set 
$A$. The overall inequality and the definition of ICWE show that $f_{r}^{(i)}=0$. \newline
\textbf{Claim 2:} For any type $i\in B$ and any route $r\in \mathcal{R}_{A}$, we have $%
\tilde{f}_{r}^{(i)}=0$, i.e., for a given route, if the cost of the route in the equilibrium increases in the game with higher information, then the equilibrium costs of
all types that are using this route in the equilibrium of the higher
information game also increases. This follows since if $r\not\in \tilde{\mathcal{R}}_{i}$, then $\tilde{f}%
_{r}^{(i)}=0$. Otherwise, $r\in \tilde{\mathcal{R}}_{i}$ which implies $r\in 
\mathcal{R}_{i}$, because $1\not\in B$ and information set of all other types
are fixed. Assuming $i\in B$ and  $r\in \mathcal{R}_A$, we have 
\begin{equation*}
\tilde{c}_{r}\ {>} \ c_{r} \ {\geq } \ c^{(i)} \ {\geq } \ \tilde{c}^{(i)},
\end{equation*}%
where the first inequality follows from the definition of the set $\mathcal{R%
}_{A}$. The second inequality follows from the
definition of ICWE. The third inequality follows from the definition of set 
$B$. The overall inequality and the definition of ICWE show that $\tilde{f}_{r}^{(i)}=0$. 
\newline
\textbf{Claim 3:} Letting $s_{A}=\sum_{r\in \mathcal{R}_{A}}%
\sum_{i=1}^{K}f_{r}^{(i)}$, $\tilde{s}_{A}=\sum_{r\in \mathcal{R}%
_{A}}\sum_{i=1}^{K}\tilde{f}_{r}^{(i)}$, $s_{B}=\sum_{r\in \mathcal{R}%
_{B}}\sum_{i=1}^{K}f_{r}^{(i)}$, and $\tilde{s}_{B}=\sum_{r\in \mathcal{R}%
_{B}}\sum_{i=1}^{K}\tilde{f}_{r}^{(i)}$, we have $\tilde{s}_{A}\leq s_{A}$
and $\tilde{s}_{B}\geq s_{B}$. \newline
This follows from Claims 1 and 2. The traffic on the routes in $\mathcal{R}_{A}$ from $f^{(1:K)}$ is $s_{A}$ which is the entire traffic demand $s_{i}$ for all 
$i\in A$ (Claim 1) and possibly some portion of the traffic demand $s_{j}$ for $j\in
B $. On the other hand, the traffic on the routes in $\mathcal{R}_{A}$ from $\tilde{f}^{(1:K)}$ is $\tilde{s}_{A}$, which contains only some portion of
the traffic demand $s_{i}$ for $i\in A$. Claim 2 implies that for all $j\in B$ the traffic demand $\tilde{s}_{j}$ is only sent on the routes in $\mathcal{R}_{B}$. This
shows that $\tilde{s}_{A}\leq s_{A}$ which in turn leads to $\tilde{s}%
_{B}\geq s_{B}$ (see Figure \ref{fig:proofMultiple} for an illustration of
the partitioning and the flows).

\begin{figure}[t]
\centering
\includegraphics[width=.6\textwidth]{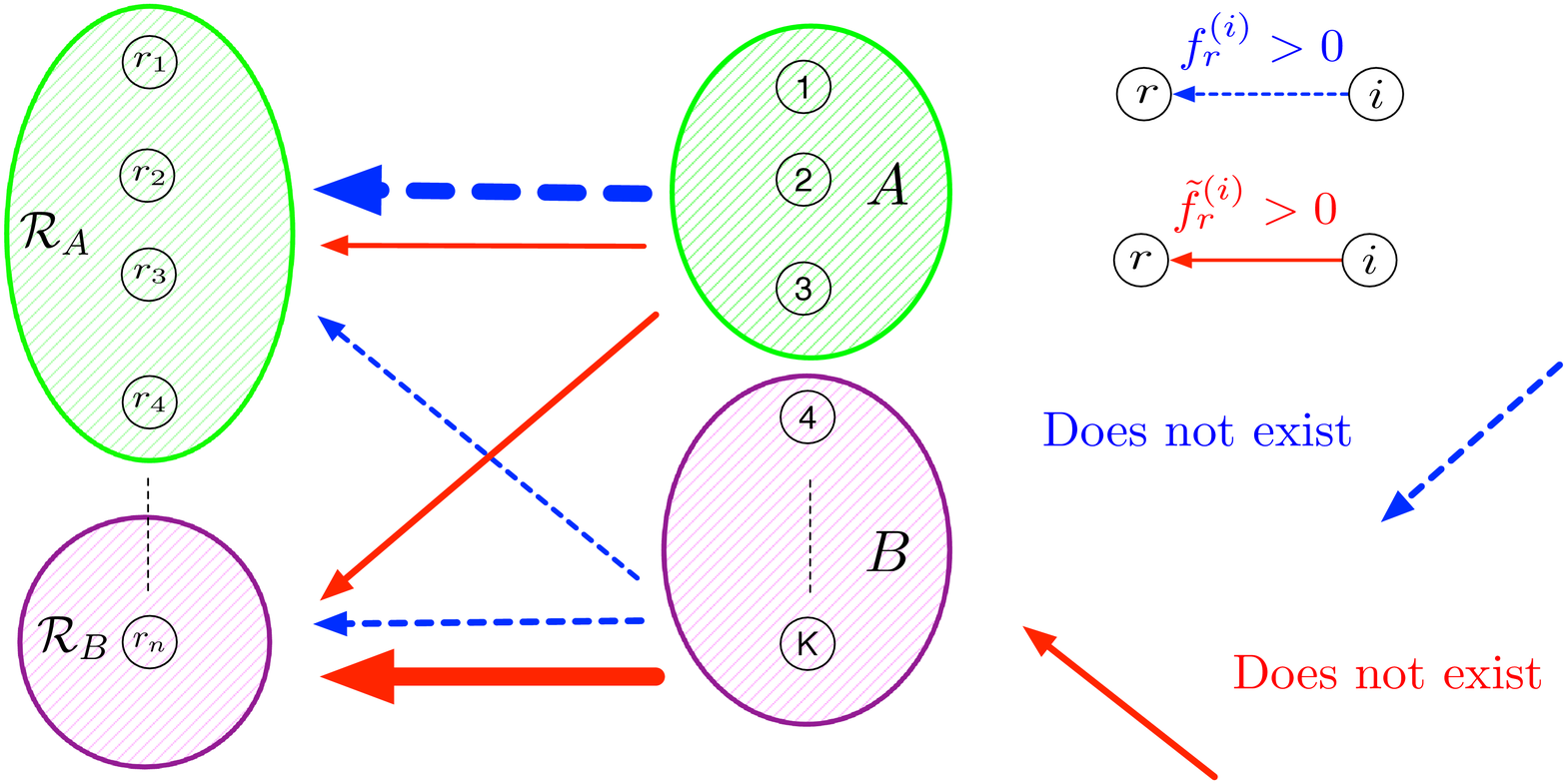}
\caption{Proof of Theorem \ref{thm:SLI}: set $A$ ($B$) represents types with higher (lower) equilibrium costs and set $\mathcal{R}_A$ ($\mathcal{R}_B$) represents routes with higher (lower) costs.  There is no dashed (blue) arrow from $A$ to $\mathcal{R}_B$ which illustrates Claim 1 and there is no solid arrow (red) from $B$ to $\mathcal{R}_A$ which illustrates Claim 2.}
\label{fig:proofMultiple}
\end{figure}
Part (b) of Lemma \ref{lemma:LI12} shows that there exists type $i$ for
which $\tilde{c}^{(i)}\leq c^{(i)}$, which in turn shows that set $B$ is
nonempty. Also by the contradiction assumption $1\in A$, which implies that both $A$ and $B$ are nonempty. Using Claim 1, if $A$ is nonempty, then $\mathcal{R}_{A}\neq \emptyset $ as the flow $f^{(1:K)}$
of the types in $A$ can only go to routes in $\mathcal{R}_{A}$. Also using Claim 2,
since $B$ is nonempty, we have $\mathcal{R}_{B}\neq
\emptyset $ as the flow $\tilde{f}^{(1:K)}$ of the types in $B$ can only go
to routes in $\mathcal{R}_{B}$. Therefore, we have partitioned the
routes of the network into two nonempty sets $\mathcal{R}_{A}$ and $%
\mathcal{R}_{B}$ such that $\tilde{c}_{r}>c_{r}$ for all $r\in \mathcal{R}%
_{A}$ and $\tilde{c}_{r} \le c_{r}$ for all $r\in \mathcal{R}_{B}$. In other
words, we have $\max_{r\in \mathcal{R}_{A}} \{ c_{r}-\tilde{c}_{r} \}<0 $
and $\min_{r\in \mathcal{R}_{B}} \{ c_{r}-\tilde{c}_{r} \} \geq 0$. We now
have all the pieces to use Lemma \ref{lem:maxminLI} which yields  
\begin{equation*}
0>\max_{r\in \mathcal{R}_{A}} \{ c_{r}-\tilde{c}_{r} \} \geq \min_{r\in 
\mathcal{R}_{B}} \{ c_{r}-\tilde{c}_{r} \} \geq 0,
\end{equation*}%
which is a contradiction, completing the proof of part (a). 
\\\textbf{Part (b):} The proof of this part follows from Theorem \ref%
{thm:SLIembedding}. 
Let $G$ be a non-SLI
network. Using Theorem \ref{thm:SLIembedding}, one of the networks shown in
Figure \ref{fig:CE2} must be embedded in $G$.
Using Example \ref{example:notSLI}, for all
networks shown in Figure \ref{fig:CE2}, there exists an assignment of cost
functions and information sets for which IBP occurs.

To construct an example
for $G$ we start from the cost functions for which the embedded network
features IBP (as shown in Example \ref{example:notSLI}) and then
following the steps of embedding, given in Definition \ref{def:embedding}, we will update the information sets as well
as the cost functions in a way that IBP occurs in the final network $G
$. The updates of information sets and cost functions are as follows.
\begin{itemize}
\item [(i)] If the step of embedding is to divide an edge, we assign half of the original edge cost to each of the new edges and update the information set by adding both
newly created edges to the same information set as of the original edge.
This guarantees that the equilibrium flow of the network after dividing an
edge is the same as the one before.
\item [(ii)] If the step of embedding is to add an edge, then we include that edge in
none of the information sets (or equivalently assign cost infinity to it). This
guarantees that the new edge is never used in any equilibrium.
\item [(iii)] If the step of embedding is to extend origin or destination, we
let the cost of the new edge to be $c(x)=x$ and update all of the information sets by adding
this edge to them. Since this edge will be used by all types and the flow on
it will not change, this step of embedding does not affect the equilibrium flow.
\end{itemize}
This construction establishes that since IBP is present in the initial
network, i.e., one of the networks shown in Figure \ref{fig:CE2}, it will
be present in the network $G$ as well. This completes the proof of part (b). %
\end{proof}
Recall that in Remark \ref{remark:expansionofexample} we showed for each of the networks shown in Figure \ref{fig:CE2} there exist infinitely many cost functions for which IBP occurs. This shows that if IBP occurs in a network, then it occurs for infinitely many cost functions. Because, if IBP occurs in a network $G$, Theorem \ref{thm:SLI} part (a) implies $G$ is not SLI and Theorem \ref{thm:SLIembedding} shows that one of the basic networks shown in Figure \ref{fig:CE2} is embedded in $G$. Finally, by construction of the proof of  Theorem \ref{thm:SLI} part (b), the cost function configuration of the basic network can be extended to network $G$, showing that IBP occurs for infinitely many cost functions. 



\subsection{IBP with Restricted Information Sets}
\label{sec:IBPres}
In this subsection, we show that restricting focus to networks with
a much more specific information structure --- whereby only one type does
not know all the edges and the change in question informs this type of all
edges ---\ allows us to establish that IBP does not occur in a larger set of
networks. Interestingly, in this case, IBP does not occur in exactly the
same set of networks on which BP does not occur, series-parallel networks,
though the two concepts continue to be very different even under this more
specific information structure. The similarity is that after the change, as
in the classic Wardrop Equilibrium setting studied for BP, there is no more
heterogeneity among users. We first define IBP with restricted information sets and then 
state the characterization of network topology which leads to it. 

\begin{definition}[\textbf{IBP with Restricted Information Sets}]
\textup{
Consider a traffic network with multiple information types $(G,%
\mathcal{E}_{1:K},s_{1:K},\mathbf{c})$. IBP with restricted information sets occurs if there exist 
expanded information sets $\tilde{\mathcal{E}}_{1:K}$ with $\mathcal{E}_{1}\subset \tilde{\mathcal{E}}_{1}=\mathcal{E}$, and $\mathcal{E}_i=\tilde{\mathcal{E}}_{i}=\mathcal{E}$ for $i=2, \dots, K$, such that the equilibrium cost of type $1$ in $(G,\tilde{\mathcal{E}}_{1:K},s_{1:K},\mathbf{c})$ is strictly larger
than the equilibrium cost of type $1$ in $(G,\mathcal{E}_{1:K},s_{1:K},
\mathbf{c})$. We denote the equilibrium cost of type $i\in [K]$ before and after expansion of information by $c^{(i)}$ and $\tilde{c}^{(i)}$, respectively. 
}
\end{definition}
\begin{theorem}
\label{thm:parttowhole}  IBP with restricted information sets does not occur if and only if the network $G$ is
SP. More specifically, we have the following.

\begin{itemize}
\item[(a)] If $G$ is SP, for any network with multiple information sets $(G,\mathcal{E}_{1:K},s_{1:K},\mathbf{c})$ with arbitrary assignment of cost functions $\mathbf{c}$, $K$, traffic demands $s_{1:K}$, and information set $\mathcal{E}_{1}$,   IBP with the restricted information sets does not occur.

\item[(b)] If $G$ is not SP, there exists an assignment of cost functions $\mathbf{c}$, $K$, traffic demands $s_{1:K}$, and information set $\mathcal{E}_{1}$ in which IBP with restricted information sets occurs.

\end{itemize}
\end{theorem}



\subsection{Extension to Multiple Origin-Destination Pairs}\label{sec:IBPmultipleODpairs}

In this subsection, we consider networks with multiple information types and
multiple origin-destination pairs as defined next.
\begin{definition}\label{def:multipleODpaitstrafficnetwork}
\textup{
Consider a graph $G= (V, \mathcal{E})$ containing $m$ origin-destination pairs denoted by $(O_{i},D_{i})$, $i \in [m]$. For any $ i \in [m]$, there are $K_{i}$ types of users, each with information set $%
\mathcal{E}_{i, j} \subseteq \mathcal{E}$, for $j\in [K_{i}]$. We refer to $(i, j)$ as the type of a user where $i \in [m]$ denotes the origin-destination pair of this type and $j \in [K_i]$ represents its information set. The traffic network with
multiple information types and multiple origin-destination pairs is denoted by $\left( G, \left\{\mathcal{E}_{i, 1:K_i} \right\}_{i=1}^{m},
\left\{s_{i, 1:K_i}\right\}_{i=1}^m, \mathbf{c} \right)$. We let $\mathcal{R}_{i, j}$ to denote the set of routes available to a user of type $(i, j)$ (i.e., routes formed by edges in $\mathcal{E}_{i, j}$). 
A feasible flow is a flow vector $f=(f^{(1, 1:K_1)}, \dots,f^{(m, 1:K_m)})$ such
that $f^{(i, 1:K_i)}$ is a feasible flow for origin-destination pair $(O_i, D_i)$. 
}
\end{definition}
We denote the total flow on an edge $e$ by $f_{e}$ where $f_{e}=\sum_{i=1}^m \sum_{j=1}^{K_i} \sum_{r \in \mathcal{R}_{i, j}~:~ e \in r}f_{r}^{(i, j)}$. Note that since $G$ is an undirected graph, the total flow on each edge is the sum of the flows sent through that edge in either direction (see \cite{lin2011stronger} and \cite{holzman2015strong}). The cost of a route $r$ is defined as $%
c_r(f)= \sum_{e \in r} c_e(f_e)$.
ICWE in this case is defined naturally as follows:

A feasible flow $f=(f^{(1, 1:K_1)},\dots ,f^{(m, 1:K_m)})$ is an Information Constrained Wardrop Equilibrium (ICWE) if for
every $i \in [m]$ and $j \in [K_i]$ and every pair $r,\tilde{r}\in \mathcal{R}_{i, j}$ with $f_{r}^{(i, j)}>0$, we have
\begin{equation}  \label{eq:defCWE}
c_{r}(f)\leq c_{\tilde{r}}(f).
\end{equation}%
This implies that all routes of type $(i, j)$ with positive flow have the same
cost, which is smaller than or equal to the cost of any other route in $\mathcal{R%
}_{i, j}$. The equilibrium cost of type $(i, j)$, denoted by $c^{(i, j)}$, is then
given by the cost of any route in $\mathcal{R}_{i, j}$ with positive flow from
type $(i, j)$.

The existence of ICWE in this setting follows from an
identical argument to that of Theorem \ref{thm:existsnceCWE}. Finally, the definition of IBP for this extended setting is as follows. 
\begin{definition}[IBP with Multiple Origin-destination Pairs]\label{def:IBPmultipleODpairs}
\textup{
Consider a traffic network with multiple information types and multiple origin-destination pairs $\left( G, \left\{\mathcal{E}_{i, 1:K_i} \right\}_{i=1}^{m},
\left\{s_{i, 1:K_i}\right\}_{i=1}^m, \mathbf{c} \right)$. IBP occurs if there exists an
expanded information set $\left\{ \tilde{\mathcal{E}}_{i, 1:K_i} \right\}_{i=1}^m$ with ${\mathcal{E}}_{1, 1} \subset \tilde{\mathcal{E}}_{1, 1}$ and $\tilde{\mathcal{E}}_{i, j}=\mathcal{E}_{i, j}$, for all $(i, j) \neq (1,1)$, $i \in [m]$, $j \in [K_i]$, such that the equilibrium cost of type $(1,1)$ in $\left( G, \left\{\tilde{\mathcal{E}}_{i, 1:K_i} \right\}_{i=1}^{m},
\left\{s_{i, 1:K_i}\right\}_{i=1}^m, \mathbf{c} \right)$ is strictly larger
than the equilibrium cost of type $(1,1)$ in $\left( G, \left\{\mathcal{E}_{i, 1:K_i} \right\}_{i=1}^{m},
\left\{s_{i, 1:K_i}\right\}_{i=1}^m, \mathbf{c} \right)$. 
}
\end{definition}
Note that the choice of type $(1,1)$ for information expansion is arbitrary and without loss of generality. We next establish a sufficient condition on the network topology
under which IBP with multiple origin-destination pairs does not arise. We will use the following definitions from \cite{chen2016network}. 
\begin{definition}
\textup{ 
\begin{itemize}
\item For any origin-destination pair $(O_i, D_i)$, the \emph{relevant network} $i$ denoted by $G_i=(V_i, \mathcal{E}_i)$ consists of all edges and nodes of $G$ that belong to at least one route from $O_i$ to $D_i$ in $G$. 
\item  For an SLI network $G_i$, each LI block has two terminal nodes, an origin and a destination, such that the origin is the first node and the destination is the last node in the block visited on any route in $G_i$. For two SLI networks $G_i$ and $G_j$, a \emph{coincident LI block} is a common LI block of $G_i$ and $G_j$ with the same set of terminal nodes, allowing origin of one to be the destination of other. 
\end{itemize} 
}
\end{definition}
Note that the definition of relevant network $G_i$ as well as its LI blocks depend only on the network $G$ and the origin-destination pair $(O_i, D_i)$, not on the information sets. Based on this definition, we next provide a sufficient condition for excluding IBP. 
\begin{proposition}\label{pro:MultipleOD}
Let $G$ be a graph with $m \ge 1$ origin-destination pairs. For any $i \in [m]$, let $G_i= (V_i, \mathcal{E}_i)$ be the relevant network for origin-destination pair $(O_i, D_i)$. IBP does not occur if the following two conditions hold. 
\begin{itemize}
\item [(a)] For any $i \in [m]$, the network $G_i$ is SLI.
\item [(b)] For any $i, i' \in [m]$ either $\mathcal{E}_i \cap \mathcal{E}_{i'} = \emptyset$ or $\mathcal{E}_i \cap \mathcal{E}_{i'}$ consists of all coincident blocks of $G_i$ and $G_{i'}$. 
\end{itemize}
\end{proposition}
\begin{proof}
We let $f$ and $\tilde{f}$ denote the equilibrium flows before and after the expansion of information set of type $(1,1)$. To reach a contradiction suppose $\tilde{c}^{(1,1)} > c^{(1,1)}$. Using Lemma \ref{lem:sumflowseries} part (b) and condition (a) of the proposition, the equilibrium cost of type $(1, 1)$ users is the sum of equilibrium cost of the LI blocks of $G_1$. Since  $\tilde{c}^{(1,1)} > c^{(1,1)}$, there exists an LI block of $G_1$ for which the equilibrium cost after expanding information set of type $(1,1)$ increases. We denote this LI block by $G^*$ and its corresponding origin and destination by $O^*$ and $D^*$. Using condition (b) for any $i \neq 1$ we have one of the following two cases: (i) $G_i$ does not have any common edge with $G^*$ and therefore none of the route flows of $(O_i, D_i)$ go through any edge of $G^*$ (ii) $O^*$ and $D^*$ belong to all routes of $G_i$ and therefore all route flows of $(O_i, D_i)$ go through $G^*$. We let $C$ be the set of indices of such origin-destination pairs, i.e., $C=\{i \in [m] ~:~ O^*, D^* \in r, \forall r \in  G_i \}$. We next define a traffic network with single origin-destination pair $(O^*, D^*)$ over $G^*$ for which IBP has occurred. The types of users are $\left( \bigcup_{i \in C} \{(i, j)~:~ j \in [K_i]\} \right) \bigcup \{(1, j)~:~ j \in [K_1] \}$ with their corresponding traffic demands.  Note that for all $i \in C$, even though our definition of coincident LI block allows the route flows of $(O_i, D_i)$ to go from $O^*$ to $D^*$ in either direction, without loss of generality, we can assume that route flows go from $O^*$ to $D^*$. This is because the cost of any edge is a function of the sum of the flows that passes through that edge in either direction and reversing the flows does not change the equilibrium flows on edges. Using Lemma \ref{lem:sumflowseries} part (a) the restriction of equilibrium flows $f$ and $\tilde{f}$ to $G^*$ are equilibrium flows for the congestion game with multiple information types and single origin-destination pair defined on $G^*$. Note that $G^*$ is an LI network and the equilibrium cost of type $(1,1)$ users after expanding their information set has gone up, which is a contradiction using Theorem \ref{thm:SLI}. 
\end{proof}
Figure \ref{fig:multipleODpairs1} shows two SLI networks with their corresponding LI blocks and Figure \ref{fig:multipleODpairs2} shows a graph with two origin-destination pairs which satisfies our sufficient condition. 

The next example shows that the conditions of Proposition \ref{pro:MultipleOD} are not necessary for non-occurrence of IBP. 
\begin{example}\label{Example:MultipleOD}
\textup{
Consider the network $G$ shown in Figure \ref{fig:ExampleMultipleOD}. The common LI block of relevant networks $G_1$ and $G_2$ is $G$ itself which is not a coincident LI block because the sets of terminals of this block for $G_1$ and $G_2$ are different. Therefore, this network does not satisfy the conditions of Proposition \ref{pro:MultipleOD}. However, in Appendix \ref{app:Example:MultipleOD} we show that for any set of edge cost functions, IBP does not occur in this network. 
}
\end{example}

\begin{figure}[t]
\centering
\begin{subfigure}[b]{0.45\textwidth}
               \includegraphics[width=\textwidth]{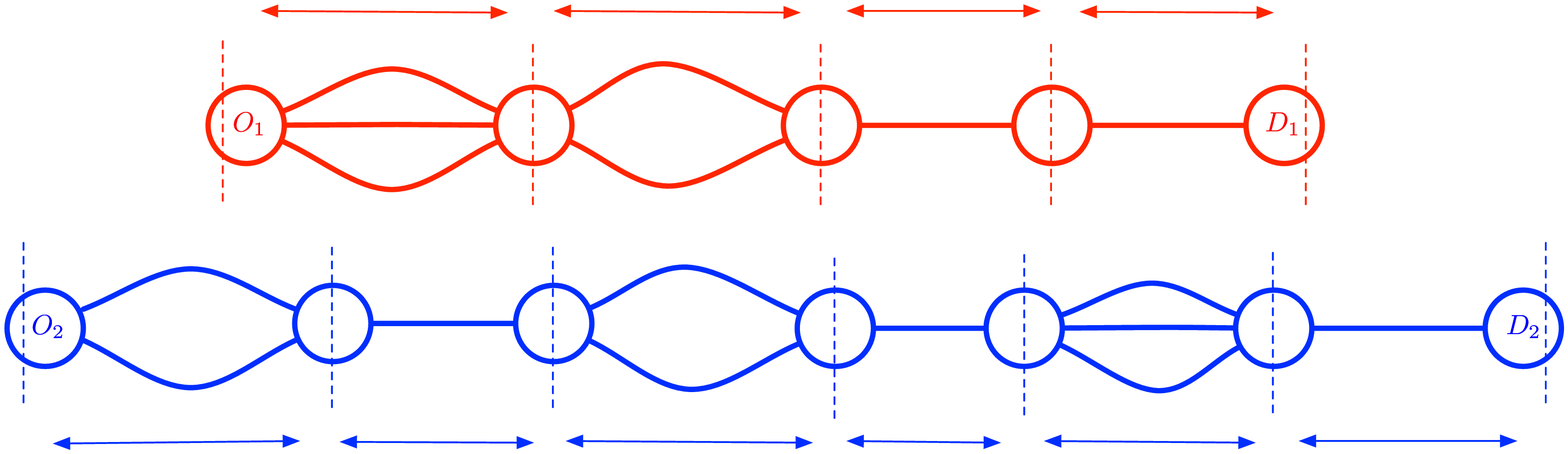}
                \caption{ }
                \label{fig:multipleODpairs1}
        \end{subfigure}
\begin{subfigure}[b]{0.35\textwidth}
 \includegraphics[width=\textwidth]{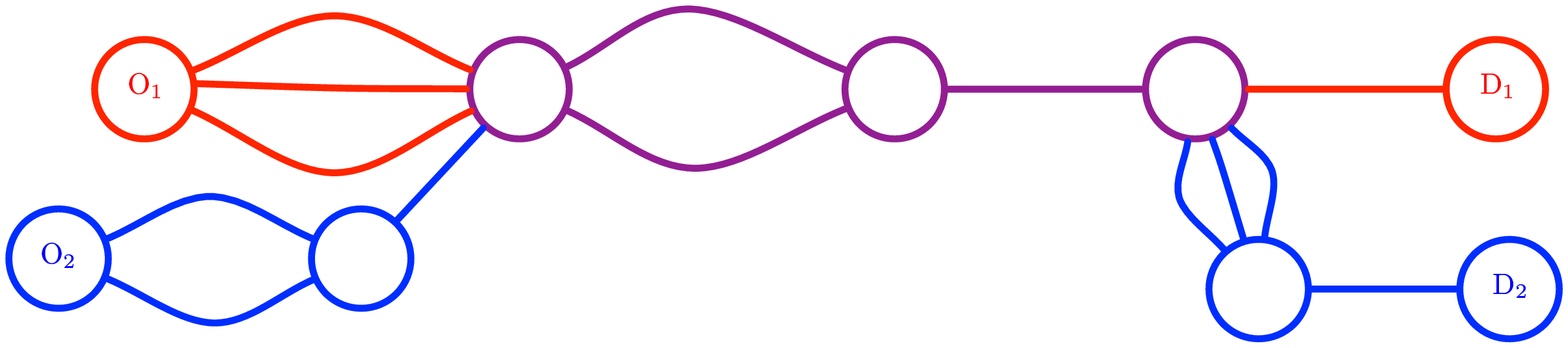}
                \caption{}
                \label{fig:multipleODpairs2}
\end{subfigure}
\caption{(a) Two SLI networks with their corresponding LI blocks. (b) A graph with two OD pairs for which IBP does not occur. }
\end{figure}

\begin{figure}[t]
\centering
               \includegraphics[width=0.35\textwidth]{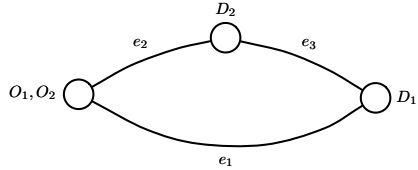}
                \caption{Example \ref{Example:MultipleOD}: IBP does not occur on this network.}
                \label{fig:ExampleMultipleOD}
\end{figure}

In concluding this subsection we should note that BP with multiple origin destination pairs has been studied in \cite{epstein2009efficient, lin2011stronger, fujishige2015matroids, holzman2015strong}, and \cite{chen2016network}. In particular, \cite{chen2016network} provide a full characterization of network topologies for which BP occurs with multiple origin-destination pairs. BP as defined in \cite{chen2016network} occurs if adding an edge (decreasing cost of an edge) increases the equilibrium cost of the users of one of the origin-destination pairs, even if that edge is never used by the users of that origin-destination pair. 
With this definition it is possible to have a network for which IBP does not occur while BP occurs. For instance, BP occurs in the network considered in Example \ref{Example:MultipleOD} (see \cite{chen2016network}) while we showed IBP does not occur in this network. 

\section{Efficiency of Information Constrained Wardrop Equilibrium}

\label{sec:POA}


In this section, we provide bounds on the inefficiency of ICWE. We show that the worst-case inefficiency remains the same as the standard Wardrop
Equilibrium, even though our notion of ICWE is considerably more general than
Wardrop Equilibrium since it allows for a rich amount of heterogeneity among
users.

We start by defining the social optimum defined as the feasible flow vector
that minimizes the total cost over all edges. We focus on aggregate efficiency loss defined as the ratio of total cost experienced by all users at social optimum and ICWE. We provide tight
bounds on this measure of efficiency loss which are realized for
different classes of cost functions. We also consider type-specific efficiency loss defined as the ratio of total cost experienced by type $i$ users at social optimum and ICWE. We show that the bounds in this case are different from the ones in the standard Wardrop Equilibrium.


Given a traffic network with multiple information types $(G,\mathcal{E}_{1:K},s_{1:K},\mathbf{c})$, we
define the social optimum, denoted by $f_{\mathrm{so}}^{(1:K)}= (f_{\mathrm{so}}^{(1)},\dots ,f_{%
\mathrm{so}}^{(K)})$ (or simply $f_{\mathrm{so}}$), as the optimal solution of the following optimization
problem: 
\begin{align}
& \min \sum_{e\in \mathcal{E}}f_{e}~c_{e}(f_{e}),  \notag
\label{eq:socialoptimum} \\
& f_{e}=\sum_{i=1}^{K}\sum_{r\in \mathcal{R}_{i}~:~e\in r}f_{r}^{(i)}, 
\notag \\
& \sum_{r\in \mathcal{R}_{i}}f_{r}^{(i)}=s_{i},\text{ and }f_{r}^{(i)}\geq 0%
\text{ for all }r\in \mathcal{R}_{i}\text{ and }i.
\end{align}%
This optimization problem minimizes the total cost over all edges incurred
by all users of all types. Under the assumption that each cost function is
continuous, it follows that the optimal solution of problem %
\eqref{eq:socialoptimum} and hence a social optimum always exists. 
We denote the total cost of a feasible flow $f^{(1:K)}$ by 
\begin{equation*}
C(f^{(1:K)})\triangleq \sum_{e\in \mathcal{E}}f_{e}c_{e}(f_{e}).
\end{equation*}%
Similarly, for a feasible flow $f^{(1:K)}$, we define the total cost
incurred by type $i$ users as 
\begin{equation*}
C^{(i)}(f^{(1:K)})\triangleq \sum_{e\in \mathcal{E}}f_{e}^{(i)}c_{e}(f_{e}).
\end{equation*}%
%
%
Consequently, we define the socially optimal cost of type $i$ as $C_{\text{so%
}}^{(i)}={C}^{(i)}(f_{\mathrm{so}}^{(1:K)})$ for $i\in[K]$ and the
overall cost (over all types) of social optimum as $C_{\mathrm{so}%
}={C}(f_{\mathrm{so}}^{(1:K)})$. Similarly, we define equilibrium cost of
type $i$ as $C_{\mathrm{cwe}}^{(i)}={C}^{(i)}(f_{\mathrm{cwe}}^{(1:K)})$ for 
$i \in [K]$ and the overall cost (over all types) of ICWE as $C_{\mathrm{%
cwe}}={C}(f_{\mathrm{cwe}}^{(1:K)})$, where $f_{\mathrm{cwe}}^{(1:K)}$ (or simply $f_{\mathrm{cwe}}$)
denotes an ICWE. Note that ${C}^{(i)}(f_{\mathrm{cwe}}^{(1:K)})$ is different
from equilibrium cost of type $i$ denoted by $c^{(i)}$, as the latter notion
is the cost per unit of flow and the former is the aggregate cost. The
relation between these two is simply $C_{\mathrm{cwe}}^{(i)}=s_{i}\ c^{(i)}$, $i\in \lbrack K \rbrack$. 

The following result from \cite{roughgarden2002bad} and \cite{correa2005inefficiency}
presents bounds on the efficiency loss of Wardrop Equilibrium, which
provides bounds on the efficiency loss of ICWE in a traffic network with
single information type with $\mathcal{E}_{1}=\mathcal{E}$ denoted by $(G,%
\mathcal{E}_{1},s_{1},\mathbf{c})$. 

\begin{proposition}[\protect\cite{roughgarden2002bad}]
\label{pro:TimPA} Consider a traffic network with a single information type $%
(G,\mathcal{E}_{1},s_{1},\mathbf{c})$. Let $f_{\mathrm{we}}$ be a Wardrop
Equilibrium and $f_{\mathrm{so}}$ be a social optimum. Then, we have

\begin{itemize}
\item[(a)] $\inf_{(G, \mathcal{E}_1,s_1, \mathbf{c}): \ c_e \  \mathrm{ convex}}
\frac{C_{\mathrm{so}}}{C_{\mathrm{we}}}=0$.

\item[(b)] Suppose $c_e(x)$ is an affine function for all $e \in \mathcal{E}$%
. Then, we have $\frac{C_{\mathrm{so}}}{C_{\mathrm{we}}} \ge \frac{3}{4}$, and
this bound is tight.

\item[(c)] Let $\mathcal{C}$ be a class of latency functions and let $\beta(%
\mathcal{C})= \sup_{c \in \mathcal{C}, \ x \ge 0} \beta(c, x)$, where
\begin{align*}
\beta(c, x)= \max_{z \ge 0} \frac{z\ (c(x)-c(z) )}{x\ c(x)}.
\end{align*}
Then, we have $\frac{C_{\mathrm{so}}}{C_{\mathrm{we}}} \ge 1- \beta(\mathcal{C})$, and
the bound is tight.
\end{itemize}
\end{proposition}

Our next result shows that Proposition \ref{pro:TimPA} holds exactly for
ICWE, indicating that within the class of heterogeneous,
information-constrained traffic equilibria we consider, the worst-case
scenario occurs for networks with homogeneous users.

\begin{proposition}
\label{pro:priceofanarchy} Consider a traffic network with multiple
information types $(G,\mathcal{E}_{1:K},s_{1:K},\mathbf{c})$. Let $f_{\mathrm{cwe}}$ be an ICWE and $f_{\mathrm{so}}$ be a social optimum. Then, we have

\begin{itemize}
\item[(a)] $\inf_{(G, \mathcal{E}_{1:K}, s_{1:K}, \mathbf{c}): ~ c_e \ \mathrm{ convex}} \frac{C_{\mathrm{so}}}{C_{\mathrm{cwe}}}=0$.

\item[(b)] Suppose $c_e(x)$ is an affine function for all $e \in \mathcal{E}$%
. Then, we have $\frac{C_{\mathrm{so}}}{C_{\mathrm{cwe}}} \ge \frac{3}{4}$, and
this bound is tight.

\item[(c)] Let $\mathcal{C}$ be a class of latency functions and let $\beta(%
\mathcal{C})= \sup_{c \in \mathcal{C}, x \ge 0} \beta(c, x)$, where
\begin{align*}
\beta(c, x)= \max_{z \ge 0} \frac{z\ (c(x)-c(z) )}{x\ c(x)}.
\end{align*}
Then, we have $\frac{C_{\mathrm{so}}}{C_{\mathrm{cwe}}} \ge 1- \beta(\mathcal{C}%
) $, and the bound is tight.
\end{itemize}
\end{proposition}
\begin{proof} We first show that for any type $i$, and any feasible flow $%
f^{(i)}$ for this type, we have 
\begin{align}  \label{eq:pro:proof:priceofanarchy}
\sum_{e \in \mathcal{E}} c_e(f_{e, \mathrm{cwe}}) ({f_{e, \mathrm{cwe}}}%
^{(i)}- f_e^{(i)}) \le 0.
\end{align}
The reason is that in ICWE each type uses only the routes with the minimal
costs. Therefore, for any type $i$ and any feasible flow $f^{(i)}$ for type $%
i$, we have 
\begin{align*}
\sum_{ r\in \mathcal{R}_i} c_r(f_{\mathrm{cwe}}^{(1:K)}) {f}^{(i)}_{r, 
\mathrm{cwe}} \le \sum_{ r\in \mathcal{R}_i} c_r(f_{\mathrm{cwe}}^{(1:K)}) {f%
}^{(i)}_r.
\end{align*}
This leads to 
\begin{align*}
0 & \ge \sum_{ r\in \mathcal{R}_i} c_r(f_{\mathrm{cwe}}^{(1:K)}) ({f}%
^{(i)}_{r, \mathrm{cwe}} - {f}^{(i)}_r) = \sum_{ r\in \mathcal{R}_i} \left(
\sum_{e: e \in r} c_e(f_{e, \mathrm{cwe}}) \right) ({f}^{(i)}_{r, \mathrm{cwe%
}} - {f}^{(i)}_r) \\
& = \sum_{e \in \mathcal{E}} c_e(f_{e, \mathrm{cwe}}) \sum_{ r\in \mathcal{R}%
_i: ~ e \in r} ({f}^{(i)}_{r, \mathrm{cwe}} - {f}^{(i)}_r)= \sum_{e \in 
\mathcal{E}} c_e(f_{e, \mathrm{cwe}}) ({f}^{(i)}_{e, \mathrm{cwe}} - {f}%
^{(i)}_e),
\end{align*}
which is the desired inequality, showing Equation %
\eqref{eq:pro:proof:priceofanarchy}. We next proceed with the proof. \newline
\textbf{Part (a)}: this holds because a traffic network with one type is a
special case of traffic network with multiple information types and part (a)
of Proposition \ref{pro:TimPA} shows the infimum is zero. \newline
\textbf{Part (b)}: using Equation \eqref{eq:pro:proof:priceofanarchy} for $f^{(i)}=
f^{(i)}_{\mathrm{so}}$ for any $i  \in [K]$, and taking summation over all
types $i \in [K]$, we obtain 
\begin{align*}
C_{\mathrm{cwe}} = \sum_{e \in \mathcal{E}} f_{e, \mathrm{cwe}} c_e(f_{e, 
\mathrm{cwe}})&= \sum_{i=1}^K \sum_{e \in \mathcal{E}} c_e(f_{e, \mathrm{cwe}%
}){f}^{(i)}_{e, \mathrm{cwe}} \le \sum_{i=1}^K \sum_{e \in \mathcal{E}}
c_e(f_{e, \mathrm{cwe}}){f}^{(i)}_{e, \mathrm{so}} \\
& = \sum_{e \in \mathcal{E}} c_e(f_{e, \mathrm{cwe}})\sum_{i=1}^K {f}%
^{(i)}_{e, \mathrm{so}} =\sum_{e \in \mathcal{E}} f_{e, \mathrm{so}}
c_e(f_{e, \mathrm{cwe}}) \\
&= \sum_{e \in \mathcal{E}} f_{e, \mathrm{so}} c_e(f_{e, \mathrm{so}}) +
\sum_{e \in \mathcal{E}} f_{e, \mathrm{so}} \left( c_e(f_{e, \mathrm{cwe}})
- c_e(f_{e, \mathrm{so}}) \right) \\
& \le \sum_{e \in \mathcal{E}} f_{e, \mathrm{so}} c_e(f_{e, \mathrm{so}}) + 
\frac{1}{4} \sum_{e \in \mathcal{E}} f_{e, \mathrm{cwe}} c_e(f_{e, \mathrm{%
cwe}}),
\end{align*}
where the last inequality comes from the fact that with $c_e(x)=a_e x +
b_e$ for $b_e, a_e \ge 0$, we have
\begin{align*}
f_{e, \mathrm{so}} \left( c_e(f_{e, \mathrm{cwe}}) - c_e(f_{e, \mathrm{so}})
\right)= a_e f_{e, \mathrm{so}} (f_{e, \mathrm{cwe}}- f_{e, \mathrm{so}})
\le \frac{1}{4} {f^2_{e, \mathrm{cwe}}} a_e \le \frac{1}{4} f_{e, \mathrm{cwe%
}} c_e(f_{e, \mathrm{cwe}}).
\end{align*}
The proof of tightness follows from part (b) of Proposition \ref%
{pro:TimPA} as a traffic network with one type is a special case of a traffic
network with multiple information types. 
\newline
\textbf{Part (c)}: using the same argument as in part (b), we obtain
\begin{align*}
C_{\mathrm{cwe}} = \sum_{e \in \mathcal{E}} f_{e, \mathrm{cwe}} c_e(f_{e, 
\mathrm{cwe}}) &\le \sum_{e \in \mathcal{E}} f_{e, \mathrm{so}} c_e(f_{e, 
\mathrm{so}}) + \sum_{e \in \mathcal{E}} f_{e, \mathrm{so}} \left( c_e(f_{e, 
\mathrm{cwe}}) - c_e(f_{e, \mathrm{so}}) \right) \\
& \le \sum_{e \in \mathcal{E}} f_{e, \mathrm{so}} c_e(f_{e, \mathrm{so}}) +
\beta(\mathcal{C}) \sum_{e \in \mathcal{E}} f_{e, \mathrm{cwe}} c_e(f_{e, 
\mathrm{cwe}}),
\end{align*}
where the last inequality comes from the fact that 
\begin{align*}
f_{e, \mathrm{so}} \left( c_e(f_{e, \mathrm{cwe}}) - c_e(f_{e,\mathrm{so}})
\right) \le \beta(c_e, f_{e, \mathrm{cwe}}) f_{e, \mathrm{cwe}} c_e(f_{e, 
\mathrm{cwe}}) \le \beta(\mathcal{C}) f_{e, \mathrm{cwe}} c_e(f_{e, \mathrm{%
cwe}}).
\end{align*}
The proof of the tightness follows from part (c) of Proposition \ref%
{pro:TimPA}.
 \end{proof}


In concluding this section, we should note that in this environment with
heterogeneous users, there are alternatives to our formulation of the social
optimum problem, which considers the \textquotedblleft
utilitarian\textquotedblright\ social optimum, summing over the costs of all
groups. An alternative would be to consider a weighted sum or focus on the
class of users suffering the greatest costs. We next illustrate that if we
focus on type-specific costs, even with {affine} cost functions, some groups
of users may have worse than $3/4$ performance relative to the social optimum. 


\begin{example}
\label{example:POA} \textup{\ Consider the network shown in Figure \ref
{fig:examplePOA} with $\mathcal{E}_{1}=\{e_{1}\}$, $\mathcal{E}%
_{2}=\{e_{1},e_{2}\}$. The ICWE is $f_{e_{1},\mathrm{cwe}}^{(1)}=s_{1}$ and $f_{e_{1},%
\mathrm{cwe}}^{(2)}=\frac{1}{a}-s_{1}$, $f_{e_{2},\mathrm{cwe}}^{(2)}=s_{2}-%
\frac{1}{a}+s_{1}$. The equilibrium costs are $C_{\mathrm{cwe}}^{(1)}=s_{1}$
and $C_{\mathrm{cwe}}^{(2)}=s_{2}$. The social optimum is ${f}_{e_{1},%
\mathrm{so}}^{(1)}=s_{1}$ and $f_{e_{1},\mathrm{so}}^{(2)}=\frac{1}{2a}-s_{1}$,
 ${f}_{e_{2},\mathrm{so}}^{(2)}=s_{2}-\frac{1}{2a}+s_{1}$. The
corresponding costs are $C_{\mathrm{so}}^{(1)}=\frac{s_{1}}{2}$ and $C_{\text{%
so}}^{(2)}=\frac{-1}{4a}+s_{2}+\frac{s_{1}}{2}$ (assuming $\frac{1}{2a}\geq
s_{1}$ and $s_{2}\geq \frac{1}{a}-s_{1}$). Therefore, we have
\begin{equation*}
\frac{C_{\mathrm{so}}^{(1)}}{C_{\mathrm{cwe}}^{(1)}}=\frac{1}{2},~\frac{C_{\text{%
so}}^{(2)}}{C_{\mathrm{cwe}}^{(2)}}=\frac{\frac{-1}{4a}+s_{2}+\frac{s_{1}}{2}}{%
s_{2}},\text{ and }\frac{C_{\mathrm{so}}^{(1)}+C_{\mathrm{so}}^{(2)}}{C_{\text{%
cwe}}^{(1)}+C_{\mathrm{cwe}}^{(2)}}=\frac{s_{1}+s_{2}-\frac{1}{4a}}{s_{1}+s_{2}%
}.
\end{equation*}%
We next show that the ratio of the aggregate costs is greater than or equal
to $3/4$. We have $s_{1}+s_{2}\geq \frac{1}{a}$ which leads to
\begin{equation*}
\frac{C_{\mathrm{so}}}{C_{\mathrm{cwe}}}=\frac{s_{1}+s_{2}-\frac{1}{4a}}{%
s_{1}+s_{2}}=1-\frac{1}{4a}\frac{1}{s_{1}+s_{2}}\geq 1-\frac{1}{4a}a=\frac{3%
}{4}.
\end{equation*}%
However, the type-specific efficiency loss can be smaller than $\frac{3}{4}$ as we
have $\frac{C_{\mathrm{so}}^{(1)}}{C_{\mathrm{cwe}}^{(1)}}<\frac{3}{4}$. }
\end{example}

\begin{figure}[t]
\centering
\includegraphics[width=0.25\textwidth]{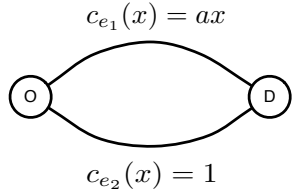}
\caption{Example \protect\ref{example:POA}, type-specific efficiency loss
versus aggregate efficiency loss.}
\label{fig:examplePOA}
\end{figure}

\section{Concluding Remarks}

\label{sec:conclusion}

GPS-based route guidance systems, such as Waze or Google maps, are rapidly
spreading among drivers because of their promise of reduced delays as they
inform their users about routes that they were not aware of or help them
choose dynamically between routes depending on recent levels of congestion.
Nevertheless, there is no systematic analysis of the implications for
traffic equilibria of additional information provided to subsets of users.
In this paper, we systematically studied this question. We first extended
the class of standard congestion games used for analysis of traffic
equilibria to a setting where users are heterogeneous because of their
different information sets about available routes. In particular, each
user's information set contains information about a subset of the edges in
the entire road network, and drivers can only utilize routes consisting of
edges that are in their information sets. We defined the notion of
Information Constrained Wardrop Equilibrium (ICWE), an extension of the
classic Wardrop Equilibrium notion, and established the existence and
essential uniqueness of ICWE.

We then turned to our main focus, which we formulate in the form of
Informational Braess' Paradox (IBP). IBP asks that whether users receiving
additional information can become worse off. Our main result is a
comprehensive answer to this question. We showed that in any network in the
Series of Linearly Independent (SLI) class, which is a strict subset of
series-parallel network, IBP cannot occur, and in any network that is not in
the SLI class, there exists a configuration of edge-specific cost functions
for which IBP will occur. The SLI class is comprised of networks that join
linearly independent networks in series, and linearly independent networks are
those for which every path between origin and destination contains at least
one edge that is not in any other such path. This is the property that
enables us to prove that IBP cannot occur in any SLI network. We also showed
that any network that is not in the SLI class necessarily embeds at least
one of a specific set of basic networks, and then used this property to show that
IBP will occur for some cost configurations in any non-SLI network. We further
proved that whether a given network is SLI can be determined in linear time.
Finally, we also established that the worst-case inefficiency performance of
ICWE is no worse than the standard Wardrop Equilibrium with one type of
users.

There are several natural research directions which are opened up by our
study. These include:

\begin{itemize}
\item Our analysis focused on the effect of additional information on the
set of users receiving the information. For what classes of networks is
additional information very harmful for other users? This question is
important from the viewpoint of fairness and other social objectives. We may
like that users utilizing route guidance systems are experiencing lower
delays, but not if this comes at the cost of significantly longer delays for
others.

\item How \textquotedblleft likely\textquotedblright\ are the cost function
configurations that will cause IBP to occur in non-SLI networks. This
question is important for determining, ex-ante before knowing the exact
traffic flows, whether additional information for some sets of users, coming
for example from route guidance systems, might be harmful.

\item Is there an \textquotedblleft optimal information\textquotedblright\
configuration for users of a traffic network? Specifically, one could
consider the following question: given the traffic demands of $K$ types, $%
s_{1},\dots ,s_{K}$, find the information sets $\mathcal{E}_{1},\dots ,%
\mathcal{E}_{K}$ that generate the minimum overall cost for all types in an
ICWE. This question is related to \cite{roughgarden2001designing,
roughgarden2006severity} who investigate the question of finding the
subnetwork of the initial network that leads to optimal equilibrium cost
with one type of user.

\item We established a sufficient condition under which IBP does not occur on a traffic network with multiple origin-destination pairs. One natural question is to find sufficient and necessary condition for this problem. 

\item Finally, our study poses an obvious empirical question, complementary
to similar studies for the Braess' Paradox: are there real-world settings
where we can detect IBP? 
\end{itemize}


\section*{Acknowledgment}
We would like to thank anonymous referees and editors as well as participants at several seminars and conferences for useful comments and suggestions. In particular, we thank Nicolas Stier, Rakesh Vohra, Ozan Candogan, and Saeed Alaei for helpful comments. 
\section{Appendix}
\subsection{Proofs of Section \ref{sec:existence}}\label{app:proofs:sec:existence}
\subsubsection{Proof of Proposition \ref{pro:potential}} 
Since for any $e \in \mathcal{E}$ the function $c_{e}(\cdot)$ is nondecreasing, $\int_{0}^{f_{e}}c_{e}(z)dz$ as a function of $%
f_{r}^{(i)}$ is convex and continuously differentiable. 
\\ \textbf{Claim 1}: If $f^{(1:K)}$ is an optimal solution of  \eqref{eq:propotential}, then it is an ICWE. 
\\ Since the objective function is convex and the constraints are affine functions, regularity conditions holds and KKT conditions are satisfied, i.e., there exists $\mu_{i, r} \le 0$ and $\lambda_i$ such that for all $ i \in [K]$ and $r\in \mathcal{R}_i$ we have
\begin{align}\label{eq:legrangiandetails1}
\frac{\partial }{\partial f_r^{(i)}} \left( \sum_{e\in \mathcal{E}}\int_{0}^{f_{e}}c_{e}(z)dz-\sum_{i=1}^{K} \lambda_i \left( \sum_{r\in \mathcal{R}_{i}}f_{r}^{(i)}-s_{i}\right) + \sum_{r, i} \mu_{r,i} f_r^{(i)}  \right)=0, 
\end{align}
where $\mu_{r, i} =0$ for $f_r^{(i)} > 0$ \citet[Chapter 3]{bertsekas1999nonlinear}. We show that the flow $f^{(1:K)}$ is an ICWE with the equilibrium cost of type $i$ being $\lambda_i$. First, note that $f^{(1:K)}$ is a feasible flow by the constraints of \eqref{eq:propotential}. Second, we can rewrite \eqref{eq:legrangiandetails1} as 
\begin{align}
\sum_{e \in \mathcal{E}} \frac{\partial f_e }{\partial f^{(i)}_r }  c_e(f_e) = \sum_{e \in \mathcal{E} ~: ~ e \in r}  c_e(f_e) = \begin{cases} = \lambda_i& \mbox {if } $$f_r^{(i)}>0, \\
\ge \lambda_i   & \mbox {if } $$f_r^{(i)}=0,
\end{cases}
\end{align}
where we used $\mu_{r, i}=0$ for $f_r^{(i)}>0$ in the first case and $\mu_{r,i} \le 0$ for $f_r^{(i)}=0$ in the second case. This is exactly the definition of ICWE which completes the proof of Claim 1. 
\\ \textbf{Claim 2}: If $f^{(1:K)}$ is an ICWE, then it is an optimal solution of \eqref{eq:propotential}. 
\\ We let the equilibrium cost of type $i$ users be $\lambda_i$ which leads to the following relation
\begin{align}
\sum_{e \in \mathcal{E} ~: ~ e \in r}  c_e(f_e) = \begin{cases} = \lambda_i & \mbox {if } $$f_r^{(i)}>0, \\
\ge \lambda_i   & \mbox {if } $$f_r^{(i)}=0.
\end{cases}
\end{align}
For all $i \in [K]$ and $r \in \mathcal{R}_i$, if $f_r^{(i)}>0$, then we define $\mu_{i, r}=0$ and if $f_r^{(i)}=0$, then we define $\mu_{i, r}= \lambda_i - \sum_{e \in \mathcal{E} ~: ~ e \in r}  c_e(f_e)$. First, note that $\mu_{i, r} \le 0$ and if $f_r^{(i)}>0$, then $\mu_{i, r}=0$. Second, note that 
\begin{align}
\frac{\partial }{\partial f_r^{(i)}} \left( \sum_{e\in \mathcal{E}}\int_{0}^{f_{e}}c_{e}(z)dz-\sum_{i=1}^{K} \lambda_{i}\left( \sum_{r\in \mathcal{R}_{i}}f_{r}^{(i)}-s_{i}\right) + \sum_{r, i} \mu_{r,i} f_r^{(i)}  \right)=0.
\end{align}
Therefore, the flow $f^{(1:K)}$ together with $\lambda_i$ and $\mu_{i,r}$ satisfy the KKT conditions. Since the objective function of \eqref{eq:propotential} is convex and the constraints are affine functions, KKT conditions are sufficient for optimality \citet[Chapter 3]{bertsekas1999nonlinear}, proving the claim. 

\subsubsection{Proof of Theorem \ref{thm:existsnceCWE}} The set of feasible flows $f^{(1:K)}$ is a compact subset of $%
K|\mathcal{R}|$-dimensional Euclidean space. Since edge cost functions are
continuous, the potential function is also continuous. Weierstrass extreme
value theorem establishes that optimization problem \eqref{eq:propotential}
 attains its minimum which by Proposition \ref{pro:potential} is an ICWE.
\newline
We next, show that in two different equilibria $f^{(1:K)}$
and $\tilde{f}^{(1:K)}$, the equilibrium cost for
each type is the same. By Proposition \ref{pro:potential}, both $%
f^{(1:K)}$ and $\tilde{f}^{(1:K)}$ are
optimal solutions of \eqref{eq:propotential}. Since $\Phi(\cdot)$ is a convex
function, we have 
\begin{align*}
& \Phi \left( \alpha f^{(1:K)} +(1-\alpha )\tilde{f}^{(1:K)} \right)  \leq \alpha \Phi \left( f^{(1:K)} \right)+(1-\alpha )\Phi \left( \tilde{f}^{(1:K)}\right),
\end{align*}%
for any $\alpha \in \lbrack 0,1\rbrack$. Since $\Phi \left( f^{(1:K)} \right)$ and
$\Phi \left( \tilde{f}^{(1:K)}\right)$ are both equal to optimal
value of \eqref{eq:propotential}, and for each $e$, the function $%
\int_{0}^{f_{e}}c_{e}(z)dz$ is convex (its derivative with respect to $f_{e}$
is $c_{e}(f_{e})$ which is non-decreasing), the functions $%
\int_{0}^{f_{e}}c_{e}(z)dz$ for any $e\in \mathcal{E}$ must be linear
between values of $f_{e}$ and $\tilde{f}_{e}$. This shows that all cost
functions $c_{e}$ are constant between $f_{e}$ and $\tilde{f}_{e}$ and in particular the equilibrium costs are the same.
\subsection{Proofs of Section \ref{sec:graphtheory}}
\subsubsection{Proof of Equivalence in  Definition \ref{def:SLI}}\label{app:proofofDefSLI}
We first show that each LI network $G$ is the result of attaching several LI blocks in series. This follows by induction on the number of edges. Using Definition \ref{def:LI}, $G$ is either the result of attaching two LI networks in parallel or the result of attaching an LI network and a single edge in series. If $G$ is the result of attaching two LI networks in parallel, then $G$ is biconnected and so is an LI block. If $G$ is the result of attaching an LI network $G_1$ with a single edge, then the single edge is an LI block and by induction hypothesis $G_1$ is series of several LI blocks. Therefore, $G$ is the result of attaching several LI blocks in series. 

We next show that the following two definitions are equivalent.
\begin{itemize}
\item An SLI network  is either a single LI network or the connection of two SLI networks in series. We let SET1 to denote the set of such networks. 
\item An SLI network consists of attaching several LI blocks in series. We let SET2 to denote the set of such networks.
\end{itemize}
We show that SET1=SET2 by induction on the number of edges, i.e., we suppose that for any network with number of edges less than or equal to $m$ these two sets are equal and then show that for networks with $m+1$ edges the two sets are equal as well (note that the base of this induction for $m=1$ corresponds to a single edge which evidently holds).
\begin{itemize}
\item If a network $G$ belongs to SET1, then either it is a single LI network or is the result of attaching two SLI networks in series. In the former case, it belongs to SET2 as we have shown each LI network is the result of attaching several LI blocks. In the latter case, by induction hypothesis both SLI subnetworks are the series of several LI blocks and so is their attachment in series. This shows SET1 $\subseteq$ SET2. 
\item If a network $G$ belongs to SET2, then either it is a single LI block or is the result of attaching several LI blocks in series. In the former case, by definition it belongs to SET1. In the latter case, we let $G_1$ to denote the LI block that contains origin and the series of the rest of LI blocks by $G_2$. $G_1$ is SLI by definition as it is a single LI block and $G_2$ is SLI by induction hypothesis. Therefore, the series attachment of $G_1$ and $G_2$ belongs to SET1. This shows SET2 $\subseteq$ SET1, completing the proof. 
\end{itemize}

\subsubsection{Proof of Theorem \ref{thm:SLIembedding}}
\begin{figure}[t]
\centering
\includegraphics[width=0.5\textwidth]{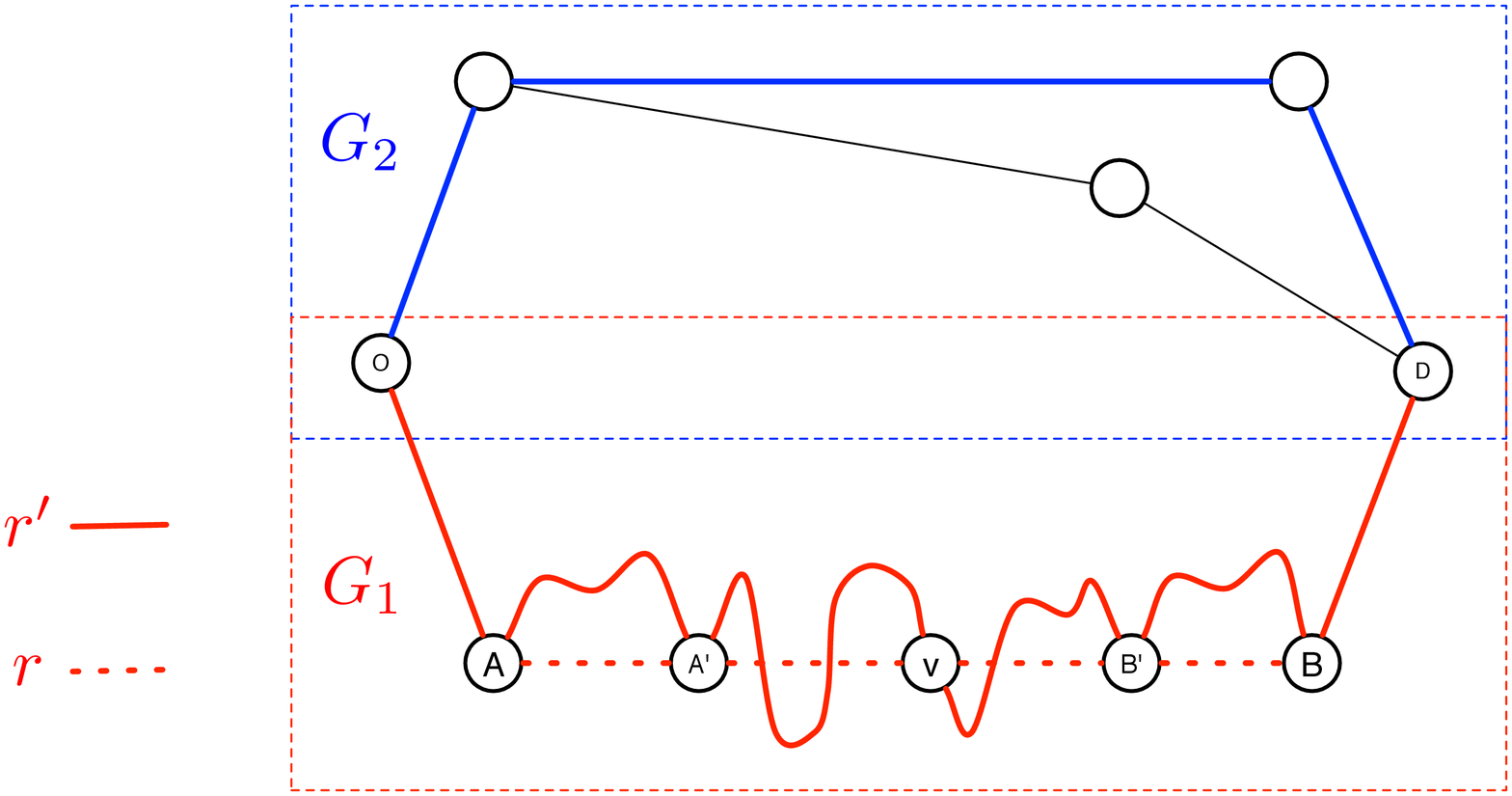}
\caption{Proof of Theorem \ref{thm:SLIembedding}: $G_{1}$ is not LI and $G_{2}$ has at
least one route from O to D. }
\label{fig:pfSLI}
\end{figure}
We first show that if a network $G$ belongs to the class SLI, then none of
the networks shown in Figure \ref{fig:CE2} is embedded in it. First note that
since all networks in the class SLI are series-parallel, using part (b) of
Proposition \ref{pro:alternatedef} implies that the Wheatstone network
shown in Figure \ref{fig:CE23} is not embedded in it. The SLI network $G$ consists of
several LI blocks that are attached in series. Using part (a) of Proposition \ref{pro:alternatedef} none of the networks shown in Figure \ref{fig:CE2} (i.e., networks 
shown in Figs. \ref{fig:CE21}, \ref{fig:CE22},  \ref{fig:CE25},  \ref{fig:CE26},  \ref{fig:CE27},  \ref{fig:CE28},  \ref{fig:CE29},  and \ref{fig:CE210}) can be embedded in one of the LI blocks. We next show that they cannot be embedded in the series of two LI blocks as well.  We let $%
G_1$ and $G_2$ be two LI blocks that are attached in series where the
resulting network from this attachment is $H$. Also, we let the node $c$ be the
attaching node of these two networks. We will show that the network shown
in Figure \ref{fig:CE21} can not be embedded in $H$ (a similar argument shows
that the rest of the networks shown in Figure \ref{fig:CE2} cannot be embedded in it). In
order to reach to a contradiction, we suppose the contrary, i.e., $H$ is
obtained from the network shown in Figure \ref{fig:CE21} by applying the
embedding procedure described in Definition \ref{def:embedding}. We define
the corresponding routes to $e_5$, $e_1 e_4$, and $e_2 e_3$ in $H$ by $r_3$,
$r_1$ and $r_2$. Formally, we start from $r_3=e_5$, $r_1= e_1 e_4$, and $%
r_2= e_2 e_3$ in the network shown in Figure \ref{fig:CE21} and at each step of the embedding procedure whenever we divide an edge on $%
r_i$ ($i=1,2,3$) we will update $r_i$ by adding that edge and whenever we
extend origin or destination we will add the new edge to all $r_i$'s. Given this
construction, in the network $H$ we have three routes $r_3$, $r_1$, and $r_2$,
where $r_1$ and $r_2$ have a common node and they do not have any common
node (except $O$ and $D$) with $r_3$. This is a contradiction as
all routes in $H$ must have node $c$ in common. This completes the proof of the first part. 
\newline
We next show that if none of networks shown in Figure \ref{fig:CE2} is embedded
in $G$, then $G$ belongs to the class SLI. Proposition \ref
{pro:alternatedef}(b) implies that since Figure \ref{fig:CE23} is not embedded
in $G$, it is series-parallel. We next show that given
a series-parallel network $G$, if $G$ is not SLI then we can find an
embedding of one of networks shown in Figures \ref{fig:CE21}, $\cdots$,\ref%
{fig:CE210} in it. The proof is by induction on the number of edges of $G$.
Following Definition \ref{def:SP}, consider the last
building step of the network $G$. If the last step, is attaching two
networks $G_1$ and $G_2$ in series, then assuming that $G$ is not SLI, we
conclude that either $G_1$ or $G_2$ (or both) is not SLI. Therefore, by
induction hypothesis, we can find an embedding of one of the networks shown in Figures \ref%
{fig:CE21},$\cdots$, \ref{fig:CE210} in either $G_1$ or $G_2$, which in turn shown it is embedded in $G$.  If the last step, is attaching two
networks $G_1$ and $G_2$ in parallel, then it must be the case that that either $G_1$ or $G_2$ is not LI. Because otherwise the parallel attachment of two LI networks is LI (Definition \ref{def:LI}) and hence SLI, which contradicts the fact that $G$ is not SLI. Without loss of generality, we let
the network that is not LI to be $G_1$. Therefore, part (a) of Proposition 
\ref{pro:alternatedef} shows that there exist two routes $r$ and $%
r^{\prime }$ and a vertex $v$ common to both routes such that both sections $%
r_{Ov}$ and $r^{\prime }_{Ov}$ as well as $r_{vD}$ and $r^{\prime }_{vD}$
are not equal (note that $v \not\in \{O, D\}$ because otherwise if $v=O$, then $r_{Ov}=r^{\prime }_{Ov}$ as both are the single node $O$).

Note that using part (b) of Proposition \ref{pro:alternatedef}, there is a way to index vertices such that along any route, the vertices have increasing indices. We let $A$ be the last vertex (with the prescribed indexing) before which the two
routes $r$ and $r^{\prime }$ become the same (this vertex can be $O$
itself). Since $v$ is the common vertex of theses two routes and $r_{Ov} \neq
r^{\prime }_{Ov}$ such a vertex exists.
Because $v$ is a common vertex of $r$ and $r^{\prime }$ the two routes $r$
and $r^{\prime }$ have a common vertex between $A$
and $v$. We let $A^{\prime }$ to be the first such vertex (it can be $v$ itself). Similarly, we define $B$ as the first
vertex after which $r$ and $r^{\prime }$ become the same ($B$ can be $D$ itself) and $%
B^{\prime }$ as the last vertex after $v$ for which $r$ and $r^{\prime}$ coincide ($B^{\prime }$ can be $v$ itself). Given these definitions for the nodes $v, A, A', B$, and $B'$, we know that $r_{AA'}$ (the path between $A$ and $A'$ on $r$) and $r^{\prime}_{AA'}$ (the path between $A$ and $A'$ on $r^{\prime}$) do not have any vertex in common and similarly $r_{BB'}$ and $r^{\prime}_{BB'}$ do not have any vertex in common. The definition of the nodes $A,
A^{\prime }, B, B^{\prime }$ is illustrated in Figure \ref{fig:pfSLI}. Next, we
show that one of the networks shown in Figures \ref{fig:CE21},$\dots$, \ref{fig:CE210} is embedded in $G$. We have the following
cases:

\begin{itemize}
\item $A=O$, $B=D$, $A^{\prime }=v$, and $B^{\prime }=v$: in this case, the
network shown in Figure \ref{fig:CE21} is embedded in $G$. This is because
there are two disjoint paths from $O$ to $v$ and from $v$ to $D$ and there
is at least one path from $O$ to $D$ in $G_2$. Since any other edge and vertex of
the network belongs to a path that connects $O$ to $D$, we can construct the
graph $G$ by starting from the network shown in Figure \ref{fig:CE21} and
applying the embedding procedure.

\item $A=O$, $B=D$, and $A^{\prime }\neq v$ or $B^{\prime }\neq v$: in this case, the network shown in Figure \ref{fig:CE22} is embedded in $G$. This is because there is at least one path from $O$ to $D$ in $G_2$ and the network shown in Figure \ref{fig:CE22} is embedded in $G_1$. To see this, note that the edges $e_1$ and $e_2$ are embedded in the section of the routes $r$ and $r^{\prime}$ between $O$ and $A^{\prime}$, and the edges $e_3$ and $e_4$ are embedded in the the section of the routes $r$ and $r^{\prime}$ between $B^{\prime}$ and $D$. Also, note that the single edge $e_6$ is embedded in the network between $A^{\prime}$ and $B^{\prime}$ (single edge is embedded in any network). 

\item $A \neq O$, $B=D$, and $A^{\prime }= v$ and $B^{\prime }= v$: the network shown in Figure \ref{fig:CE25} is embedded in $G$.

\item $A = O$, $B \neq D$, and $A^{\prime }= v$ and $B^{\prime }= v$: the network shown in Figure \ref{fig:CE26} is embedded in $G$.

\item $A = O$, $B \neq D$, and $A^{\prime }\neq v$ or $B^{\prime }\neq v$: the network shown in Figure \ref{fig:CE27} is embedded in $G$.

\item $A \neq O$, $B = D$, and $A^{\prime }\neq v$ or $B^{\prime }\neq v$: the network shown in Figure \ref{fig:CE28} is embedded in $G$.

\item $A \neq O$, $B \neq D$, and $A^{\prime }\neq v$ or $B^{\prime }\neq v$%
: the network shown in figure \ref{fig:CE29} is embedded in $G$.

\item $A \neq O$, $B \neq D$, and $A^{\prime }= v$ and $B^{\prime }= v$: the network shown in Figure \ref{fig:CE210} is embedded in $G$.
\end{itemize}
This completes the proof.

\endproof

\subsubsection{Proof of Proposition \ref{pro:recognitionSLI}}
We use the following results and definitions in this proof. 
\begin{proposition}[\protect\cite{valdes1979recognition}]
\label{pro:treedecomposition}
A network is series-parallel if following the steps $S$ and $P$ shown in Figure \ref{fig:tarjaninduction} in any order, turns the network into a single
edge connecting origin to destination. 
Moreover, if a network is series-parallel, then in linear time $O(|\mathcal{E}|+ |V|)$ we can obtain a binary tree decomposition (shown in Figure \ref
{fig:treedecomposition}) which indicates a sequence of $S$ and $P$ that turns $G$ into a single edge. 
\end{proposition}
\begin{figure}[t]
\centering
\includegraphics[width=0.3\textwidth]{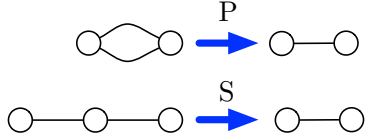}
\caption{Two operations that turns a series-parallel network into a single
edge.}
\label{fig:tarjaninduction}
\end{figure}
\begin{figure}[t]
\centering
\includegraphics[width=0.6\textwidth]{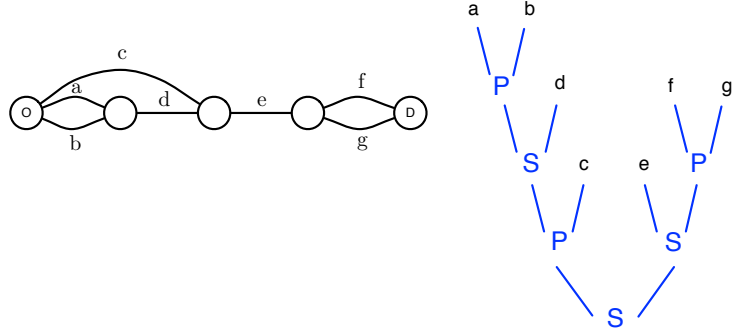}
\caption{Binary tree decomposition of a series-parallel network.}
\label{fig:treedecomposition}
\end{figure}
We now proceed with the proof of Proposition \ref{pro:recognitionSLI}. Using Proposition \ref{pro:treedecomposition}, we first verify whether $G$
is series-parallel or not, which can be done in linear time. If $G$ is not series-parallel, then it is not SLI as well. If $G$
is series-parallel, then a binary tree decomposition can be obtained in linear time (again using
Proposition \ref{pro:treedecomposition}). Note that the binary tree decomposition is not unique and the following argument works with any binary tree decomposition. In this tree the edges of $G$ are represented by the leaves of the tree. We label the incident edges to origin by $O$ and the incident edges to destination by $D$ (an edge might be labeled both $O$ and $D$). Since $G$ is SP, by definition it is the result of attaching two SP networks in series or parallel. If it is the result of attaching two SP networks in series, then there exist a node of the tree labeled $S$, referred to as the root of the tree, such that on one of the subtrees starting from that node we only have $O$ labeled leaves and on the other subtree we only have $D$ labeled leaves (this can be done in linear time by traversing the tree). If $G$ is the result of attaching two SP networks in parallel, then there exists a node of the tree labeled $P$, again referred to as the root of the tree, such that on both subtrees starting from it we have both $O$ and $D$ labeled leaves.

We next show by induction on the size of tree that whether the binary tree represents an SLI network can be verified in linear time. If the root of the tree is $S$, then we have
series of two networks. By induction hypothesis in linear time we can verify whether each of these subtrees represent and SLI network, which in turn determines whether $G$ is SLI. If the root of the tree is $P$, we
need to check whether each subtree represents an LI network. We next show this can be done in linear time which concludes the proof.
\\ \textbf{Claim:} Given the binary tree decomposition, we can verify whether the underlying network is LI in linear time.
\\  We show this claim by induction on the size of the tree as well. Starting from the root of the tree, if the root has label $P$, then by
induction hypothesis, for each of the subtrees denoted by $T_1$ and $T_2$, we can verify whether the underlying network is LI in $
O(V_{T_1})$ and $O(V_{T_2})$, respectively. The underlying network is LI if and only if both of these subtrees represent an LI network. Therefore, in $O(V)$ it can be verified whether the underlying network is
LI. If the root is labeled $S$, then the underlying network is LI if and only if one of the subtrees is only labeled $S$, and the other subtree is LI. Using any traversing algorithm (breadth first search, or depth first
search), one can visit all nodes in both subtrees in linear time, verifying if it only has $S$
labels. Furthermore, by induction we can verify whether each subtree
represents an LI network. Therefore, in linear time, we can verify whether the network is LI, completing the proof.

\subsection{Proofs of Section \ref{sec:IBP}}
\subsubsection{Expansion of Example \ref{example:notSLI}}\label{app:exampleinfinitelymany}
We provide the example for part (a) of Example \ref{example:notSLI}.  Let $K=1$, $c_{e_1}(x)=x, c_{e_2}(x)=1$, $c_{e_3}(x)=1, c_{e_4}(x)=x, c_{e_5}(x)=0$ and $s_1=1$. Also, we let the information sets be $\mathcal{E}_1=\{e_1, e_2, e_3, e_4\}$ and $\tilde{\mathcal{E}}_1=\{e_1, e_2, e_3, e_4, e_5 \}$. In equilibrium, we have $f^{(1)}_{e_1e_3}=f^{(1)}_{e_2e_4}=\frac{1}{2}$ with $c^{(1)}=\frac{3}{2}$ and $\tilde{f}^{(1)}_{e_1e_3}=\tilde{f}^{(1)}_{e_2e_4}=0$, $\tilde{f}^{(1)}_{e_1e_5e_4}=1$ with $\tilde{c}^{(1)}=2$. Therefore, after expanding the information set of type $1$ users their equilibrium cost has increased from $\frac{3}{2}$ to $2$. 
\subsubsection{Proof of the Claim of Remark \ref{remark:expansionofexample}}\label{app:expansionofexample}
We will show that there are infinitely many cost functions for the network shown in Figure \ref{fig:CE21} for which IBP occurs. In particular, we show the following claim.  
\\ \textbf{Claim:} For any $a_1, a_3, a_5 > 0$ such that $a_1+a_3 > a_5$, there exist non-negative $ b_1, b_2, b_3, b_4, b_5, a_2, s_1$, and $s_2$ such that with cost functions $c_{e_i}(x)=a_i x+ b_i$, $1 \le i \le 5$, IBP occurs in the network shown in Figure \ref{fig:CE21}. In particular, we show the following cost function parameters along  with $\mathcal{E}_2=\{e_1, e_4, e_5\}$, $\mathcal{E}_1=\{e_2, e_3, e_5\}$, and $\tilde{\mathcal{E}}_1=\{e_1, e_2, e_3, e_5\}$ leads to IBP.  
\begin{align*}
& a_4=b_1=b_3=b_5=0, \\
& b_2 = a_1 y = a_1  \frac{a_5 (s_1 + s_2)}{a_1+a_3+a_5}, \\
& b_4= \frac{a_5 a_3 (s_1+s_2)}{a_1+a_3+a_5}, \\
& \frac{s_1}{s_1+s_2} \in \left(\frac{a_1+a_3}{a_1+a_3+a_5}, \min \left\{ \frac{(a_3+ a_5) (a_3a_5 + a_1^2 + a_1a_3 + a_1 a_5)- a_1a_5^2}{(a_1+a_3+a_5)\left(a_3 a_5 + a_3a_1 +a_1 a_5 \right)}, 1 \right\} \right), \\
& a_2= \frac{a_5^2 \left(\frac{s_1}{s_1+s_2} \left( a_1+a_3+a_5 \right) - a_1 \right)}{a_5\left(a_1+a_3+a_5 \right) \frac{s_1}{s_1+s_2}- a_1\left(a_1+a_3+a_5 \right)\frac{s_2}{s_1+s_2}-a_3a_5} - a_3 - a_5.
\end{align*}
\\ \emph{Proof:} we let $a_4=b_1=b_3=b_5=0$ and then find $a_2, b_2, b_4, s_1$, and $s_2$ for which IBP occurs with $\mathcal{E}_2=\{e_1, e_4, e_5\}$, $\mathcal{E}_1=\{e_2, e_3, e_5\}$, and $\tilde{\mathcal{E}}_1=\{e_1, e_2, e_3, e_5\}$. We will find the $a_2, b_2, b_4, s_1$, and $s_2$ parameters such that before expanding the information set, the equilibrium flow is $f^{(2)}_{e_5}=0$, $f^{(2)}_{e_1 e_4}=s_2$, and $f^{(1)}_{e_5}=s_1-x$, $f^{(1)}_{e_2 e_3}=x$. We will further impose the constraint that the cost of route $e_5$ for type $2$ users is equal to the cost of route $e_1 e_4$. For this to hold it is sufficient and necessary to have $a_5 (s_1-x)= a_2 x + b_2 + a_3 x$, which leads to 
\begin{align}\label{eq:exampleeq1}
x= \frac{a_5 s_1 - b_2}{a_2+ a_3+ a_5} \in [0, s_1].
\end{align}
We also have $a_1 s_2 + b_4 = a_5 (s_1-x)$, which leads to 
\begin{align}\label{eq:exampleeq2}
a_1 s_2 + b_4= a_5 \left(s_1 - \frac{a_5 s_1 - b_2}{a_2+ a_3+ a_5}  \right)
\end{align}
We will also choose $a_2, b_2, b_4, s_1$, and $s_2$ parameters such that after expanding the information set, the equilibrium flow becomes $\tilde{f}^{(2)}_{e_5}=s_2$, $\tilde{f}^{(2)}_{e_1 e_4}=0$, $\tilde{f}^{(1)}_{e_5}=s_1-y$, $\tilde{f}^{(1)}_{e_2 e_3}=0$, and $\tilde{f}^{(1)}_{e_1 e_3}=y$. We will further impose the constraint that the cost of all available routes for each type of users are equal. For this to hold it is sufficient and necessary to have $a_5 (s_1 + s_2 - y) = a_1 y + a_3 y$, which leads to 
\begin{align}\label{eq:exampleeq4}
y= \frac{a_5 (s_2 + s_1)}{a_1+a_3+a_5} \in (0, s_1).
\end{align}
We also have $a_1 y + a_3 y = b_2 + a_3 y$, which after substituting $y$ from \eqref{eq:exampleeq4} leads to 
\begin{align}\label{eq:exampleeq5}
b_2 = a_1 y = a_1  \frac{a_5 (s_2 + s_1)}{a_1+a_3+a_5}.
\end{align}
Also, for type $2$ users we have $a_1 y + b_4 = a_5 (s_2 + s_1 - y)$, which after substituting $y$ from \eqref{eq:exampleeq4} leads to 
\begin{align}\label{eq:exampleeq3}
b_4= \frac{a_5 a_3 (s_2+s_1)}{a_1+a_3+a_5}.
\end{align}
Therefore, Equations \eqref{eq:exampleeq3} and \eqref{eq:exampleeq5} determine $b_2$ and $b_4$ as a function of other parameters. In what follows we will show how to choose non-negative $s_1, s_2$, and $a_2$ such that Equations \eqref{eq:exampleeq1}, \eqref{eq:exampleeq2}, and \eqref{eq:exampleeq4} hold as well. After some rearrangements, we can see that the constraints imposed by Equations \eqref{eq:exampleeq1} and \eqref{eq:exampleeq4} are equivalent to 
\begin{align}\label{eq:exampleeqprime14}
\frac{s_1}{s_2+s_1} \ge \frac{\max\{a_5, a_1\}}{a_1+a_3+a_5}. 
\end{align}
Furthermore, IBP occurs if we have $a_1 y + b_4 > a_1 s_2 + b_4$, which leads to $\frac{s_2}{s_2+s_1} < \frac{a_5}{a_1 + a_3+ a_5}$ or equivalently 
\begin{align}\label{eq:exampleeqprimeIBP}
\frac{s_1}{s_2+s_1} > \frac{a_1+a_3}{a_1 + a_3+ a_5}. 
\end{align}
Using $a_1+a_3> a_5$, Equations \eqref{eq:exampleeqprime14} and \eqref{eq:exampleeqprimeIBP} become equivalent to 
\begin{align}\label{eq:exampleeqprimeprime14IBP}
\frac{s_1}{s_2 + s_1} > \frac{a_1+a_3}{a_1 + a_3+ a_5}. 
\end{align}
Using Equation \eqref{eq:exampleeq2}, we can find $a_2$ as follows
\begin{align}\label{eq:exampleeqprime2}
a_2= \frac{a_5^2 \left(\left(\frac{s_1}{s_2 + s_1} (a_1+a_3+a_5) \right) - a_1 \right)}{a_5 \left( \frac{s_1}{s_2 + s_1} (a_1+a_3+a_5)\right) - a_5 a_3 + a_1 \left( \frac{s_1}{s_2 + s_1} (a_1+a_3+a_5) \right) - a_1(a_1+a_3+a_5)} - a_3 - a_5,
\end{align}
with the condition that the right-hand side of Equation \eqref{eq:exampleeqprime2} is non-negative. From \eqref{eq:exampleeqprimeprime14IBP} the non-negativity of $a_2$ becomes equivalent to 
\begin{align}\label{eq:exampleeqprimeprim2}
\frac{s_1}{s_1 + s_2} (a_1+a_3+a_5)  \le \frac{(a_3+ a_5) (a_3a_5 + a_1^2 + a_1a_3 + a_1 a_5)- a_1a_5^2}{(a_3+a_5)(a_5+a_1)-a_5^2}.
\end{align}
Choosing $\frac{s_1}{s_1 + s_2} (a_1+a_3+a_5)$ which satisfies both Equations \eqref{eq:exampleeqprimeprim2} and \eqref{eq:exampleeqprimeprime14IBP} is feasible if we have 
\begin{align*}
a_1+a_3 < \frac{(a_3+ a_5) (a_3a_5 + a_1^2 + a_1a_3 + a_1 a_5)- a_1a_5^2}{(a_3+a_5)(a_5+a_1)-a_5^2},
\end{align*}
which after simplification becomes equivalent to $a_3 a_5^2 > 0$, and therefore holds. Hence, by choosing $\frac{s_1}{s_1 + s_2} (a_1+a_3+a_5)$ such that 
\begin{align}
\frac{s_1}{s_1 + s_2} (a_1+a_3+a_5) \in \left(a_1+a_3, \min \left\{ \frac{(a_3+ a_5) (a_3a_5 + a_1^2 + a_1a_3 + a_1 a_5)- a_1a_5^2}{(a_3+a_5)(a_5+a_1)-a_5^2}, a_1+a_3+a_5 \right\} \right),
\end{align}
all the conditions are satisfied and IBP occurs in this network for infinitely many cost functions.

\subsection{Proofs of Section \ref{sec:characIBP}}
\subsubsection{Proof of Lemma \ref{lemma:LI12}}\label{app:Proof:lemma:LI12}
Given the feasible flow $f^{(1:K)}$ for $(G,\mathcal{E}_{1:K},s_{1:K},\mathbf{c})$ we construct a feasible flow $f$ with load $\sum_{i=1}^K s_i$ for a single type of users by letting $f_r = \sum_{i=1}^K f_r^{(i)}$. Using this constructions, from two feasible flows $f^{(1:K)}$ and $\tilde{f}^{(1:K)}$ we obtain two feasible flows $f$ and $\tilde{f}$ for a single type congestion game such that the load of $f$ is larger than or equal to the traffic demand of $\tilde{f}$. Therefore, part (a) follows from \citet[Lemma 5]{milchtaich2006network}. 

We next show part (b). Since part (a) holds for any two feasible flows, we
can apply it for the equilibrium flows $f^{(1:K)}$ and $\tilde{f}^{(1:K)}$
over the traffic networks $(G, \mathcal{E}_{1:K}, s_{1:K}, \mathbf{c})$ and $%
(G, \tilde{\mathcal{E}}_{1:K}, s_{1:K}, \mathbf{c})$, respectively (we can
view $f^{(1:K)}$ as a feasible flow over the traffic network $(G, \tilde{%
\mathcal{E}}_{1:K}, s_{1:K}, \mathbf{c})$ as well). It follows that there
exists a route $r$ such that $\sum_{i=1}^K f^{(i)}_r > \sum_{i=1}^K \tilde{f}%
^{(i)}_r$ and $f_e \ge \tilde{f}_e$ for all $e \in r$. From the first
inequality it follows that $\sum_{i=1}^K f^{(i)}_r > 0$ which shows at least
one of the types, say type $i$, sends a positive traffic on route $r$. Note
that $i$ can be any element of $[K]$ (it can also be $1$ as the
flow $f^{(1:K)}$ is a feasible flow for the traffic network $(G, \mathcal{E%
}_{1:K}, s_{1:K}, \mathbf{c})$). We obtain
\begin{align*}
c^{(i)}=c_r \ge \tilde{c}_r \ge \tilde{c}^{(i)},
\end{align*}
where the first equality follows from $f^{(i)}_r > 0$. The first
inequality follows from $f_e \ge \tilde{f}_e$ for all $e \in r$. The
second inequality follows from the definition of ICWE and the fact that if
type $i$ users can use route $r$ in $(G, \mathcal{E}_{1:K}, s_{1:K}, \mathbf{%
c})$, then they can use it in $(G, \tilde{\mathcal{E}}_{1:K}, s_{1:K},
\mathbf{c})$ as well since the information sets are not smaller in the
second game. This completes the proof.
%
\subsubsection{Proof of Lemma \ref{lem:sumflowseries}}\label{app:lem:sumflowseries}
\textbf{Part (a)}: Suppose $f^{(1:K)}$ is an equilibrium flow on $G$. We next show that the restriction of $f^{(1:K)}$ to $G_1$ creates an equilibrium for $G_1$. Consider type $i$ users and let $r_1$ be a route in $G_1$ such that $f^{(i)}_{r_1} > 0$ and let $r'_1$ be another route in $G_1$ which belongs to the information set of type $i$ users. The route $r_1$ is part of a route $r$ in $G$ for which $f^{(i)}_r > 0$. We let $r_2$ be the restriction of $r$ to $G_2$ (so that $r= r_1+ r_2$). Since $f^{(1:K)}$ is an equilibrium of $G$, we have  
$c_{r}=c_{r_1}+ c_{r_2} \le c_{r'_1} + c_{r_2}= c_{r_1'+r_2}$ which leads to $c_{r_1} \le c_{r'_1}$, showing that the restriction of $f^{(1:K)}$ to $G_1$ is an equilibrium. Similarly, the restriction to $G_2$ is an equilibrium. 
\\ \textbf{Part (b):} We consider an equilibrium $f^{(1:K)}$ for $G$ and then using part (a) we consider the equilibria of $G_1$ and $G_2$ obtained by restriction of $f^{(1:K)}$ to $G_1$ and $G_2$. For a type $i$ and route $r$ such that $f^{(i)}_r > 0$, we have $c^{(i)}= c_r= c_{r_1}+ c_{r_2}$, where $r_1$ and $r_2$ are the restriction of $r$ to $G_1$ and $G_2$, respectively (note that the only common node of $r_1$ and $r_2$ is the destination of $G_1$ which is the same as the origin of $G_2$, hence the operation $r_1+r_2$ is a valid operation). Since $f^{(i)}_{r_1} > 0$ and $f^{(i)}_{r_2} > 0$ we have $c_{r_1}= c^{(i)}_1$ and $c_{r_2} = c^{(i)}_2$, which leads to $c^{(i)}=c^{(i)}_1+c^{(i)}_2$.

\subsubsection{Proof of Theorem \ref{thm:parttowhole}}\label{app:proof:thm:parttowhole}
We first show two lemmas that we will use in the proof. The first lemma directly follows from the results of \cite{milchtaich2006network} for single type congestion game. 

\begin{lemma}
\label{lemma:SProuteexist} Consider a traffic network with multiple
information types $(G, \mathcal{E}_{1:K}, s_{1:K}, \mathbf{c})$. Let $f^{(1:K)}$ and $\tilde{f}^{(1:K)}$ be two (arbitrary) non-identical feasible flows such that $\sum_{i=1}^{K}s_{i}\geq
\sum_{i=1}^{K}\tilde{s}_{i}$.  If $G$ is
series-parallel, there exists a
route $r$ such that $f_e \ge \tilde{f}_e$ and $f_e > 0$ for all $e \in r$.
\end{lemma}

\begin{proof}
 Similar to the proof of Lemma \ref{lemma:LI12}, given a feasible flow $f^{(1:K)}$ for $(G,\mathcal{E}_{1:K},s_{1:K},\mathbf{c})$ we define a feasible flow $f$ with traffic demand $\sum_{i=1}^K s_i$ for a congestion game with a single information type. Therefore, this lemma follows from \citet[Lemma 2]{milchtaich2006network}. 
\end{proof}

\begin{lemma}
\label{lem:costbetween} Consider a traffic network with multiple information
types $(G, \mathcal{E}_{1:K}, s_{1:K}, \mathbf{c})$ where $\mathcal{E}_i=\mathcal{E}$ for $i=2, \dots, K$, $\mathcal{E}_1 \subseteq \mathcal{E}$, and $G$ is a SP network. Consider an ICWE with flow $(f^{(1)}, \dots,
f^{(K)})$ and let $r$ be a route for which $f_e > 0$ for any $e \in r$. We
have 
\begin{align*}
c_r \in \left[ \min_{i \in [K]} c^{(i)}, \max_{i \in [K]} c^{(i)} \right],
\end{align*} 
where for any $i \in [K]$, $c^{(i)}$ denotes the equilibrium cost of type $i$ users. 
\end{lemma}
\begin{proof} Since all the types except type $1$ have full information, we have $c^{(i)}=c^{(j)}$, for all $i, j \in \{2, \dots, K\}$, $\max_{i \in [K]} c^{(i)}= c^{(1)}$, and $\min_{i \in [K]} c^{(i)}= c^{(j)}$, $j\neq 1$. By definition of ICWE, we have $c_r \ge c^{(i)}$ (as $r \in \mathcal{R}_i$) for all $i \ge 2$. This leads to $c_r \ge \min_{i \in [K]} c^{(i)}$, showing the lower bound. We will next show the upper bound. 
We will prove this by induction on the number of edges of $G$. It evidently holds for a single edge as all equilibrium costs are equal to $c_r$. We next show the result for a series-parallel network $G$.  Since $G$ is SP, it is either the result of attaching two
SP networks in series or attaching two SP networks
in parallel. If $G$ is the result
of attaching two SP networks $G_A$ and $G_B$ in series, then using part (a) of Lemma \ref{lem:sumflowseries}, an
ICWE for the overall network is obtained by concatenating an ICWE for $G_A$
with an ICWE for $G_B$. We let $r_A$ and $r_B$ denote the sections of $r$ that belong to $G_A$ and $G_B$, respectively. We also let $c^{(i)}_A$ and $c^{(i)}_B$ be the equilibrium costs of type $i$ users in $G_A$ and $G_B$, respectively. By induction hypothesis, we have $c_{r_A} \le \max_{i \in [K]} c^{(i)}_A$ and $c_{r_B} \le  \max_{i \in [K]} c^{(i)}_B$. Since the traffic demands of type $1$ users on both $G_A$ and $G_B$ are non-zero, we have $\max_{i \in [K]} c^{(i)}_A= c^{(1)}_A$ and $\max_{i \in [K]} c^{(i)}_B= c^{(1)}_B$. 
This leads to 
\begin{align*}
c_{r} = c_{r_A} + c_{r_B} \le c^{(1)}_A + c^{(1)}_B = c^{(1)},
\end{align*}
where we used part (b) of Lemma \ref{lem:sumflowseries} in the last equality. 

Now suppose that $G$ is the result of attaching $G_A$ and $G_B$ in parallel
and suppose $r \in G_A$. Let $T=\{i \ge 2 ~:~ f^{(i)}_A > 0\}$ denotes the set of types that are sending a non-zero flow over $G_A$. Depending on whether $T = \emptyset$, we have the following two cases:
\begin{itemize}
\item $T = \emptyset$: since $f_e > 0$ for all $e \in r$, at least one type must send a non-zero flow over $G_A$ and since $T = \emptyset$, only type $1$ sends a non-zero flow over $G_A$. Therefore, we have $c_r= c^{(1)}$. We also have $c^{(1)} \le \max_{i \in [K]} c^{(i)}$, leading to $c_r \le \max_{i \in [K]} c^{(i)}$. 
\item $T \neq \emptyset$: we either have $f^{(1)}_A > 0$ or $f^{(1)}_A = 0$. If $f^{(1)}_A > 0$, then by induction hypothesis, we have 
\begin{align*}
c_r \le \max_{i \in T \cup\{1\}} c^{(i)}_A = \max_{i \in T \cup\{1\}} c^{(i)} \le \max_{i \in [K]} c^{(i)},
\end{align*}
where the equality holds because each type $i \in T \cup\{1\}$ sends a positive flow over $A$ and its equilibrium cost in $G$ is the same as its equilibrium cost in $G_A$. If $f^{(1)}_A = 0$, then again by induction hypothesis and using $T \neq \emptyset$, we have 
\begin{align*}
c_r \le \max_{i \in T} c^{(i)}_A = \max_{i \in T} c^{(i)} \le \max_{i \in [K]} c^{(i)},
\end{align*}
where the equality holds because each type $i \in T $ sends a positive flow over $A$. 
\end{itemize}

This concludes the proof of lemma. 
\end{proof}

\textbf{Proof of part (a) of Theorem \ref{thm:parttowhole}:} After expanding
 information set of type $1$ users to $\mathcal{E}$, we obtain $\tilde{c}^{(i)}=%
\tilde{c}^{(1)}$ for all $i\in \lbrack K \rbrack$. Using Lemma \ref
{lemma:SProuteexist}, there exists a route $r$ such that $f_{e}\geq \tilde{f}%
_{e}$ and $f_{e}>0$ for any $e\in r$. We have 
\begin{equation*}
c_{r}\geq \tilde{c}_{r}\geq \tilde{c}^{(i)} = \tilde{c}^{(1)}, \quad \forall ~i \in [K], 
\end{equation*}%
where the first inequality follows form $f_{e}\geq \tilde{f}_{e}$, the
second inequality follows from the definition of ICWE, and the
equality follows from $\tilde{\mathcal{E}}_{i}=\mathcal{E}$ for all $%
i=1,\dots ,K$. Since $\mathcal{E}_{1}\subseteq \mathcal{E}$, we have $%
c^{(i)}=c^{(j)}\leq c^{(1)}$ for all $i,j=2,\dots ,K$. Using Lemma \ref
{lem:costbetween}, this leads to
\begin{equation*}
c_{r}\le \max_{i \in [K]} c^{(i)}=c^{(1)}.
\end{equation*}%
Combining the previous two relations leads to $\tilde{c}^{(1)}\leq c^{(1)}$%
. 
\newline
\textbf{Proof of part (b) of Theorem \ref{thm:parttowhole}:} The proof is similar to the proof of part (b) of Theorem \ref{thm:SLI}. 
In Example \ref
{example:notSLI} we have provided an example showing that IBP with restricted information sets can occur over
Wheatstone network shown in Figure \ref
{fig:CE23}.  \newline
Suppose that a network $G$ is not series-parallel. Using Proposition \ref
{pro:alternatedef}, $G$ can be constructed from Wheatstone network shown in Figure \ref
{fig:CE23} by following the steps of embedding. 
To construct an example for $G$, we start from the cost functions for which the embedded network
features IBP with restricted information sets and then
following the steps of embedding we will update the information sets as well
as the cost functions in a way that IBP occurs in the final network which is $G
$. The updates are identical to those described in the proof of part (b) 
of Theorem \ref{thm:SLI} and establishes that if IBP with restricted information sets is present in the initial
network (i.e., the Whetstone network shown in Figure \ref{fig:CE23}), it will
be present in network $G$ as well. This completes the proof of part (b). 
\subsubsection{Omitted Proof of Example \ref{Example:MultipleOD}}\label{app:Example:MultipleOD}
First, note that after expansion of information, without loss of generality, each type of users $(i, j)$, $i=1, 2$ have at least two routes from $O_i$ to $D_i$. Because, otherwise, if a type with traffic demand $s$ has only one route $r$, we can consider an equivalent game in which we update the cost of all edges on $r$ from $c_e(x)$ to $c_e(x+s)$. Also, note that due to symmetry we can only consider the information expansion of one of the types of the form $(1, j)$. Therefore, without loss of generality we assume that there exists one type from $O_2$ to $D_2$ with information about all edges of the network and there exist either one or two types from $O_1$ to $D_1$. Below, we examine all possible cases and show that IBP does not occur:
\begin{itemize}
\item [(1)] There exist two types $\{(1,1), (2,1)\}$ such that $\mathcal{R}_{2,1}=\{e_1e_3, e_2\}$, $\mathcal{R}_{1,1}=\{e_1\}$, and $\tilde{\mathcal{R}}_{1,1}=\{e_1, e_2e_3\}$: If type $(1,1)$ does not use route $e_2 e_3$ after information expansion, then equilibrium remains the same. Now suppose, type $(1,1)$ uses route $e_2 e_3$ (i.e., $\tilde{f}^{(1,1)}_{e_1} < f^{(1,1)}_{e_1}= s_{1,1}$). If $\tilde{f}_{e_1} \le f_{e_1}$, then we have 
\begin{align*}
\tilde{c}^{(1,1)} \le c_{e_1}(\tilde{f}_{e_1}) \le c_{e_1}(f_{e_1}) = c^{(1,1)},
\end{align*}
which shows IBP does not occur. Now suppose $\tilde{f}_{e_1} > f_{e_1}$, which in turn shows $\tilde{f}_{e_2} < f_{e_2}$ as $\tilde{f}_{e_1} +\tilde{f}_{e_2}=  f_{e_1}+ f_{e_2}= s_{1,1}+ s_{2,1}$. We have 
\begin{align*}
\tilde{f}_{e_3} = \tilde{f}^{(1,1)}_{e_2 e_3} + \tilde{f}^{(2,1)}_{e_1e_3} >  f^{(1,1)}_{e_2 e_3} + f^{(2,1)}_{e_1e_3}= f_{e_3},
\end{align*}
where we used $\tilde{f}^{(1,1)}_{e_2 e_3} > f^{(1,1)}_{e_2 e_3}=0$ and $\tilde{f}^{(2,1)}_{e_1e_3} \ge f^{(2,1)}_{e_1e_3}$, which holds because $\tilde{f}^{(2,1)}_{e_1e_3}= \tilde{f}_{e_1}- \tilde{f}^{(1,1)}_{e_1} > f_{e_1}- s_{1,1}=f_{e_1}- f^{(1,1)}_{e_1}=  f^{(2,1)}_{e_1e_3}$. Therefore, we have 
\begin{align}\label{eq:tempEx01}
c_{e_1}(\tilde{f}_{e_1}) + c_{e_3}(\tilde{f}_{e_3}) \le c_{e_2}(\tilde{f}_{e_2}) \le c_{e_2}(f_{e_2}) \le c_{e_1}(f_{e_1}) + c_{e_3}(f_{e_3}),
\end{align}
where the first inequality holds because $\tilde{f}^{(2,1)}_{e_1e_3} > f^{(2,1)}_{e_1e_3} \ge 0$, the second inequality holds because $\tilde{f}_{e_2} < f_{e_2}$, and the third inequality holds because $f^{(2,1)}_{e_2}= s_{2,1}- f^{(2,1)}_{e_1e_3} > s_{2,1}- \tilde{f}^{(2,1)}_{e_1e_3}=  \tilde{f}^{(2,1)}_{e_2}\ge 0$. Inequality \eqref{eq:tempEx01} together with $c_{e_1}(\tilde{f}_{e_1}) \ge c_{e_1}(f_{e_1})$ and $c_{e_3}(\tilde{f}_{e_3}) \ge c_{e_3}(f_{e_3})$ shows that the cost of all three edges before and after information expansion are the same, leading to the same equilibrium cost for all types. Therefore, IBP does not occur in this case. 
\item [(2)] There exist two types $\{(1,1), (2,1)\}$ such that $\mathcal{R}_{2,1}=\{e_1e_3, e_2\}$, $\mathcal{R}_{1,1}=\{e_2e_3\}$, and $\tilde{\mathcal{R}}_{1,1}=\{e_1, e_2e_3\}$: If type $(1,1)$ does not use $e_1$ after information expansion, then equilibrium remains the same. Now suppose type $(1,1)$ uses route $e_1$. We show IBP does not occur in this case by considering all possibilities as follows:
\begin{itemize}
\item $\tilde{f}_{e_1} \le f_{e_1}$: Since $f_{e_1}+ f_{e_2}= s_{1,1}+ s_{2,1}= \tilde{f}_{e_1}+ \tilde{f}_{e_2}$, we have $\tilde{f}_{e_2} \ge f_{e_2}$. We also have $\tilde{f}_{e_3} < f_{e_3}$, because
\begin{align*}
\tilde{f}_{e_3} = \tilde{f}^{(1,1)}_{e_2e_3} + \tilde{f}^{(2,1)}_{e_1 e_3} <  f^{(1,1)}_{e_2e_3} + f^{(2,1)}_{e_1 e_3} = f_{e_3},
\end{align*}
where we used $\tilde{f}^{(1,1)}_{e_2e_3} < f^{(1,1)}_{e_2e_3} $ as type $(1,1)$ is using $e_1$ after information expansion and $\tilde{f}^{(2,1)}_{e_1 e_3} < f^{(2,1)}_{e_1 e_3}$ as $\tilde{f}^{(2,1)}_{e_1 e_3}= \tilde{f}_{e_1}- \tilde{f}^{(1,1)}_{e_1} < f_{e_1} = f^{(2,1)}_{e_1 e_3}$. The inequality $\tilde{f}^{(2,1)}_{e_1 e_3}< f^{(2,1)}_{e_1 e_3}$ implies $\tilde{f}^{(2,1)}_{e_2} > f^{(2,1)}_{e_2} \ge 0$. 
Therefore, we have 
\begin{align}\label{eq:tempEx1}
c_{e_2}(\tilde{f}_{e_2}) \le c_{e_1}(\tilde{f}_{e_1})+ c_{e_3}(\tilde{f}_{e_3}) \le  c_{e_1}(f_{e_1})+ c_{e_3}(f_{e_3}) \le c_{e_2}(f_{e_2}),
\end{align}
where the first inequality follows from $\tilde{f}^{(2,1)}_{e_2} > 0$, the second inequality follows from $\tilde{f}_{e_1} \le f_{e_1}$ and $\tilde{f}_{e_3} \le f_{e_3}$, and the third inequality follows from $f^{(2,1)}_{e_1e_3} > \tilde{f}^{(2,1)}_{e_1e_3} \ge 0$. Inequality \eqref{eq:tempEx1} leads to 
\begin{align*}
c^{(1,1)}=c_{e_2} (\tilde{f}_{e_2}) + c_{e_3} (\tilde{f}_{e_3}) \le c_{e_2} (f_{e_2}) + c_{e_3} (f_{e_3}) \le \tilde{c}^{(1,1)},
\end{align*}
showing IBP does not occur. 
\item $\tilde{f}_{e_1} > f_{e_1}$: Since $f_{e_1}+ f_{e_2}= s_{1,1}+ s_{2,1}= \tilde{f}_{e_1}+ \tilde{f}_{e_2}$, we have $\tilde{f}_{e_2} < f_{e_2}$. If $\tilde{f}_{e_3} \le f_{e_3}$, then we have
\begin{align*}
c^{(1,1)}=c_{e_2} (\tilde{f}_{e_2}) + c_{e_3} (\tilde{f}_{e_3}) \le c_{e_2} (f_{e_2}) + c_{e_3} (f_{e_3}) \le \tilde{c}^{(1,1)},
\end{align*}
showing IBP does not occur. Otherwise, we have $\tilde{f}_{e_3} > f_{e_3}$. First note that if $\tilde{f}_{e_3} > f_{e_3}$, then $\tilde{f}^{(2,1)}_{e_1 e_3} > f^{(2,1)}_{e_1 e_3} $. This inequality holds because
\begin{align*}
\tilde{f}^{(2,1)}_{e_1 e_3} = \tilde{f}_{e_3} - \tilde{f}^{(1,1)}_{e_2e_3} > f_{e_3} - f^{(1,1)}_{e_2e_3} = f^{(2,1)}_{e_1 e_3},
\end{align*}
where we used $\tilde{f}_{e_3} > f_{e_3}$ and $\tilde{f}^{(1,1)}_{e_2e_3} < f^{(1,1)}_{e_2e_3}$ as $(1,1)$ uses $e_1$ after the expansion of information (i.e., $f^{(1,1)}_{e_2e_3} = s_{1,1}$ and $\tilde{f}^{(1,1)}_{e_2e_3} < s_{1,1}$). Therefore, we have 
\begin{align}\label{eq:tempEx2}
c_{e_1}(\tilde{f}_{e_1}) + c_{e_3}(\tilde{f}_{e_3}) \le  c_{e_2}(\tilde{f}_{e_2}) \le  c_{e_2}(f_{e_2}) \le  c_{e_1}(f_{e_1}) + c_{e_3}(f_{e_3}),
\end{align}
where the first inequality holds because $\tilde{f}^{(2,1)}_{e_1e_3} > f^{(2,1)}_{e_1e_3} \ge 0$, the second inequality holds because $\tilde{f}_{e_2} < f_{e_2}$, and the third inequality holds because $f^{(2,1)}_{e_2}= s_{2,1}- f^{(2,1)}_{e_1e_3} > s_{2,1}- \tilde{f}^{(2,1)}_{e_1e_3}= \tilde{f}^{(2,1)}_{e_2} \ge 0 $. Inequality \eqref{eq:tempEx2} together with $\tilde{f}_{e_1}> f_{e_1}$ and $\tilde{f}_{e_3} > f_{e_3}$ leads to $c_{e_1}(f_{e_1})= c_{e_1}(\tilde{f}_{e_1})$, $c_{e_2}(f_{e_2})= c_{e_2}(\tilde{f}_{e_2})$, and $c_{e_3}(f_{e_3})= c_{e_3}(\tilde{f}_{e_3})$. Therefore, we have 
\begin{align*}
\tilde{c}^{(1,1)} \le c_{e_2}(\tilde{f}_{e_2}) + c_{e_3}(\tilde{f}_{e_3}) =  c_{e_2}(f_{e_2}) + c_{e_3}(f_{e_3}) =  c^{(1,1)},
\end{align*}
showing IBP does not occur. 
\end{itemize}
\item [(3)] There exist three types $\{(1,1), (1,2), (2,1)\}$ such that $\mathcal{R}_{2,1}=\{e_1e_3, e_2\}$, $\mathcal{R}_{1,2}=\{e_1, e_2e_3\}$, $\mathcal{R}_{1,1}=\{e_1\}$, and $\tilde{\mathcal{R}}_{1,1}=\{e_1, e_2e_3\}$: This case is similar to the first case. If type $(1,1)$ does not use $e_2e_3$ after the expansion of information, then equilibrium remains the same. Now suppose, type $(1,1)$ uses route $e_2 e_3$ (i.e., $\tilde{f}^{(1,1)}_{e_1} < f^{(1,1)}_{e_1}= s_{1,1}$). If $\tilde{f}_{e_1} \le f_{e_1}$, then we have 
\begin{align*}
\tilde{c}^{(1,1)} \le c_{e_1}(\tilde{f}_{e_1}) \le c_{e_1}(f_{e_1}) = c^{(1,1)},
\end{align*}
which shows IBP does not occur. Now suppose $\tilde{f}_{e_1} > f_{e_1}$, which in turn shows $\tilde{f}_{e_2} < f_{e_2}$ as $\tilde{f}_{e_1} +\tilde{f}_{e_2}=  f_{e_1}+ f_{e_2}= s_{1,1}+ s_{1,2}+ s_{2,1}$. We consider the following two cases:
\begin{itemize}
\item $\tilde{f}^{(1,1)}_{e_1}+ \tilde{f}^{(1,2)}_{e_1} <  f^{(1,1)}_{e_1}+ f^{(1,2)}_{e_1}$ (note that $f^{(1,1)}_{e_1}= s_{1,1}$): we have 
\begin{align*}
& \tilde{f}^{(2,1)}_{e_1e_3}= \tilde{f}_{e_1}- \left( \tilde{f}^{(1,1)}_{e_1}+ \tilde{f}^{(1,2)}_{e_1} \right) > f_{e_1} - \left( f^{(1,1)}_{e_1}+ f^{(1,2)}_{e_1} \right)= f^{(2,1)}_{e_1e_3},\\
& \tilde{f}^{(1,1)}_{e_2e_3}+ \tilde{f}^{(1,2)}_{e_2 e_3}= s_{1,1}+ s_{1,2}-  \left( \tilde{f}^{(1,1)}_{e_1}+ \tilde{f}^{(1,2)}_{e_1} \right)>  s_{1,1}+ s_{1,2}-  \left( f^{(1,1)}_{e_1}+ f^{(1,2)}_{e_1} \right)= f^{(1,2)}_{e_2 e_3}.
\end{align*}
These two inequalities lead to 
\begin{align*}
\tilde{f}_{e_3} = \tilde{f}^{(2,1)}_{e_1e_3} + \tilde{f}^{(1,1)}_{e_2e_3}+ \tilde{f}^{(1,2)}_{e_2 e_3} > f^{(2,1)}_{e_1e_3}+ f^{(1,2)}_{e_2 e_3} = f_{e_3}. 
\end{align*}
Therefore, we have 
\begin{align}\label{eq:tempEx3}
c_{e_1}(\tilde{f}_{e_1}) + c_{e_3}(\tilde{f}_{e_3}) \le c_{e_2}(\tilde{f}_{e_2}) \le c_{e_2}(f_{e_2}) \le c_{e_1}(f_{e_1}) + c_{e_3}(f_{e_3}),
\end{align} 
where the first inequality holds because $\tilde{f}^{(2,1)}_{e_1e_3} > f^{(2,1)}_{e_1e_3} \ge 0$, the second inequality holds because $\tilde{f}_{e_2} < f_{e_2}$, and the third inequality holds because $f^{(2,1)}_{e_2}= s_{2,1}- f^{(2,1)}_{e_1e_3} > s_{2,1}- \tilde{f}^{(2,1)}_{e_1e_3}=  \tilde{f}^{(2,1)}_{e_2}\ge 0$. Inequality \eqref{eq:tempEx3} together with $\tilde{f}_{e_1}> f_{e_1}$ and $\tilde{f}_{e_3} > f_{e_3}$ leads to $c_{e_1}(f_{e_1})= c_{e_1}(\tilde{f}_{e_1})$, $c_{e_2}(f_{e_2})= c_{e_2}(\tilde{f}_{e_2})$, and $c_{e_3}(f_{e_3})= c_{e_3}(\tilde{f}_{e_3})$. Therefore, the cost of all three edges before and after information expansion are the same, leading to the same equilibrium cost for all types. Therefore, IBP does not occur in this case. 
\item $\tilde{f}^{(1,1)}_{e_1}+ \tilde{f}^{(1,2)}_{e_1} \ge  f^{(1,1)}_{e_1}+ f^{(1,2)}_{e_1}$: We have 
\begin{align}
& \tilde{f}^{(1,2)}_{e_1}=  \left( \tilde{f}^{(1,1)}_{e_1}+ \tilde{f}^{(1,2)}_{e_1}\right) -  \tilde{f}^{(1,1)}_{e_1} \ge \left( f^{(1,1)}_{e_1}+ f^{(1,2)}_{e_1} \right) - \tilde{f}^{(1,1)}_{e_1}>  \left( f^{(1,1)}_{e_1}+ f^{(1,2)}_{e_1} \right) - f^{(1,1)}_{e_1} =f^{(1,2)}_{e_1}, \label{eq:tempEx4} \\
& \tilde{f}^{(1,1)}_{e_2e_3}+ \tilde{f}^{(1,2)}_{e_2 e_3}= s_{1,1}+ s_{1,2}-  \left( \tilde{f}^{(1,1)}_{e_1}+ \tilde{f}^{(1,2)}_{e_1} \right) \le  s_{1,1}+ s_{1,2}-  \left( f^{(1,1)}_{e_1}+ f^{(1,2)}_{e_1} \right)= f^{(1,2)}_{e_2 e_3}. \label{eq:tempEx5}
\end{align}
If $\tilde{f}^{(2,1)}_{e_1e_3} \le f^{(2,1)}_{e_1e_3}$, then Inequality \eqref{eq:tempEx5} leads to 
\begin{align*}
\tilde{f}_{e_3} = \tilde{f}^{(1,1)}_{e_2e_3}+ \tilde{f}^{(1,2)}_{e_2 e_3} + \tilde{f}^{(2,1)}_{e_1e_3} \le f^{(1,2)}_{e_2 e_3}+ f^{(2,1)}_{e_1e_3} = f_{e_3}. 
\end{align*}
Therefore, we obtain 
\begin{align}\label{eq:tempEx6}
c_{e_1}(\tilde{f}_{e_1}) = c_{e_2}(\tilde{f}_{e_2})+ c_{e_3}(\tilde{f}_{e_3}) \le c_{e_2}(f_{e_2})+ c_{e_3}(f_{e_3}) \le c_{e_1}(f_{e_1}),
\end{align}
where the first equality holds because using Inequality \eqref{eq:tempEx4} we obtain $\tilde{f}^{(1,2)}_{e_1} > f_{e_1}^{(1,2)} \ge 0$ and $\tilde{f}^{(1,1)}_{e_2e_3} > 0$, the first inequality holds because $\tilde{f}_{e_2} < f_{e_2}$ and $\tilde{f}_{e_3}< f_{e_3}$, and the second inequality holds because using Inequality \eqref{eq:tempEx4} we obtain $f^{(1,2)}_{e_2e_3}= s_{1,2}- f^{(1,2)}_{e_1} > s_{1,2}- \tilde{f}^{(1,2)}_{e_1}\ge 0$. Inequality \eqref{eq:tempEx6} leads to
\begin{align*}
\tilde{c}^{(1,1)} \le c_{e_1}(\tilde{f}_{e_1}) \le c_{e_1}(f_{e_1}) = c^{(1,1)},
\end{align*}
showing IBP does not occur in this case. 
\\ Now suppose $\tilde{f}^{(2,1)}_{e_1e_3} > f^{(2,1)}_{e_1e_3}$ which leads to 
\begin{align}\label{eq:tempEx8}
c_{e_1}(\tilde{f}_{e_1}) \le c_{e_1}(\tilde{f}_{e_1}) + c_{e_3}(\tilde{f}_{e_3}) \le c_{e_2}(\tilde{f}_{e_2}) \le c_{e_2}(f_{e_2}) \le c_{e_2}(f_{e_2}) + c_{e_3}(f_{e_3}) \le c_{e_1}(f_{e_1}),
\end{align}
where the second inequality holds because $\tilde{f}^{(2,1)}_{e_1e_3} > f^{(2,1)}_{e_1e_3} \ge 0$, the third inequality holds because $\tilde{f}_{e_2} < f_{e_2}$, and the last inequality holds because using Inequality \eqref{eq:tempEx4} we obtain $f^{(1,2)}_{e_2e_3}= s_{1,2}- f^{(1,2)}_{e_1} > s_{1,2}- \tilde{f}^{(1,2)}_{e_1}\ge 0$. Inequality \eqref{eq:tempEx8} leads to
\begin{align*}
\tilde{c}^{(1,1)} \le c_{e_1}(\tilde{f}_{e_1}) \le c_{e_1}(f_{e_1}) = c^{(1,1)},
\end{align*}
showing IBP does not occur in this case. 
\end{itemize}
\item [(4)] There exist three types $\{(1,1), (1,2), (2,1)\}$ such that $\mathcal{R}_{2,1}=\{e_1e_3, e_2\}$, $\mathcal{R}_{1,2}=\{e_1, e_2e_3\}$, $\mathcal{R}_{1,1}=\{e_2e_3\}$, and $\tilde{\mathcal{R}}_{1,1}=\{e_1, e_2e_3\}$: This case is similar to the second case. 
\end{itemize}

\bibliographystyle{plainnat}

\end{document}